\documentclass[journal]{IEEEtran}
\IEEEoverridecommandlockouts
\usepackage{graphicx, subcaption} 

\usepackage{amssymb}
\usepackage{latexsym}
\usepackage{soul}  
\usepackage{xcolor}  
\newcommand{\bl}[1]{\textcolor{black}{#1}}  

\usepackage{mathtools, cuted}
\usepackage{lipsum, color}
\usepackage[font=small]{caption}
\usepackage{amsmath}
\usepackage{relsize}

\usepackage{listings}
\usepackage{textgreek}
\usepackage{amsthm}
\usepackage{amssymb}
\usepackage{mdframed}
\usepackage[bookmarks=false]{hyperref}
\theoremstyle{definition}

\usepackage{diagbox}
\usepackage{leftidx}
\usepackage{algpseudocode}
\usepackage{algorithm}
\usepackage{algorithmicx}
\usepackage{mdframed}
\usepackage{amsmath,newtxtext,newtxmath}
\algrenewcommand\algorithmicrequire{\textbf{Input:}}
\algrenewcommand\algorithmicensure{\textbf{Output:}}

\algnewcommand{\IfThenElse}[3]{
  \State \algorithmicif\ #1\ \algorithmicthen\ #2\ \algorithmicelse\ #3}

\algnewcommand{\IfThen}[2]{
  \State \algorithmicif\ #1\ \algorithmicthen\ #2\ }
\usepackage[linguistics]{forest}
\usepackage{xurl}
\newcommand{\comment}[1]{}

\newtheorem{theorem}{Theorem}
\newtheorem{lemma}{Lemma}

\newtheorem{definition}{Definition}

\newtheorem{proposition}{Proposition}
\newtheorem{corollary}{Corollary}
\newtheorem{example}{Example}

\usepackage{mdframed}

\newmdtheoremenv{problem_stmt}{Problem}

\newcommand{\floor}[1]{\left\lfloor #1 \right\rfloor}

\usepackage{cite}
\usepackage{amsmath,amssymb,amsfonts}
\usepackage{graphicx}
\usepackage{textcomp}
\usepackage{xcolor}
\def\BibTeX{{\rm B\kern-.05em{\sc i\kern-.025em b}\kern-.08em
    T\kern-.1667em\lower.7ex\hbox{E}\kern-.125emX}}

\begin{document}
\bstctlcite{IEEEexample:BSTcontrol}
\title{
Low-Latency Spatial-Provenance Recovery Methods for Privacy-Constrained Vehicular Networks\\
}

\author{\IEEEauthorblockN{Manish Bansal$^{*}$ and J. Harshan$^{*, \dagger}$,\\
$^{*}$Bharti School of Telecommunication Technology and Management, $^{\dagger}$Department of Electrical Engineering,\\
Indian Institute of Technology Delhi, India
}\\
}
\maketitle

\begin{abstract}
In multihop Vehicle-to-Everything (V2X) networks, Road Side Units (RSUs) intend to collect information on vehicles' location in a low-latency manner while respecting their privacy constraints to support real-time location-based services. To facilitate data collection, provenance is known to ensure trust and accountability of data. Although existing joint data- and spatial-provenance techniques preserve the privacy of vehicles up to a certain granularity with respect to the RSU and other vehicles, they are unsuitable when stringent deadlines are imposed on the end-to-end delay on the packets. As a consequence, there is a need for designing spatial-provenance methods for V2X networks that satisfy stringent deadlines on the end-to-end delays while managing the privacy concerns. To fill this research gap, we propose two novel protocols, namely: Bi-Segment Embedding (BSE) and Tri-Segment Embedding (TSE), which provide a skipping mechanism for joint data- and spatial-provenance while trading off privacy features among the vehicles. Through an extensive theoretical framework, we provide an analysis of the proposed schemes in terms of reliability, privacy, and communication overhead. When compared to the baselines, our protocols offer lower end-to-end delay, higher reliability in provenance reconstruction, and the same level of privacy with respect to the RSU.
We validate latency gains using practical radio parameters, and our study reveals that our proposed protocols offer significant benefits in latency when implemented over a 5G stack.\looseness=-1  
\end{abstract}


\section{Introduction}
In Vehicle-to-Everything (V2X) networks, low-latency communication is essential to meet the delay deadlines on the mission-critical data that significantly affects road safety and traffic efficiency \cite{v2x}. These networks facilitate inter-vehicle communication (V2V) and vehicle-to-infrastructure communication (V2I) via Road Side Units (RSUs), which function similarly to base stations in cellular networks. Since mission-critical data is transmitted across these networks utilizing wireless technology, they are also vulnerable to security risks \cite{v2x_attacks,survey_wireless_security}. As a result, future V2X networks must implement effective techniques for monitoring each data packet for detecting potential security threats on them \cite{vanetapp}. Towards monitoring various attributes of the data packet, such as its origin and path taken by the packet, which are referred to as provenance, numerous provenance recovery techniques have been suggested in the literature to diagnose security vulnerabilities in the network \cite{provenance_survey_1,provenance_survey_2}.\looseness=-1

Provenance recovery techniques have been broadly classified into the categories of data-provenance and spatial-provenance. Data-provenance techniques  \cite{Data_provenance_survey, Data_provenance_in_vehicle_data_chains,sultana_2011,sultana_2015, Dictionary_based_provenance,IoV_data_plausibility, porkodi_secure_2021, RPL_Based_provenance_IOT, provenance_compression_IoT, multi-hop_povenance_for_IoT, amogh, suraj} focus on capturing the data-flow log of the packet through the network, whereas spatial-provenance techniques \cite{5G_multipath, mutipath_localization_1, V2V_assisted_coperative_localization, STAMP, LPPS_2, ManishTDSC, manish_WCNC, manish_COMSNETS} aim to learn the geographical positions of the wireless nodes that forwarded the packet in a network. In V2X networks, data-provenance is important to learn the logical flow of packets through vehicles, such as identifying the identity of the packet forwarders and the order in which vehicles transmit the packet to the RSU. Whereas, the localization of vehicles is crucial for the RSU to deliver location-based services (LBS) and ensure the smooth operation of autonomous vehicles. From a security perspective, vehicle localization is essential for identifying threat vectors from targeted attacks, assuring a resilient and secure vehicular network.\looseness=-1

A majority of existing methods either focus solely on data provenance or exclusively capture spatial provenance; however, they do not capture both. Identifying the absence of a method that jointly captures the data- and spatial-provenance, \cite{ManishTDSC} recently proposed a network-level protocol to achieve these objectives jointly in multihop V2X networks. In particular, \cite{ManishTDSC} considered a multihop V2X network where vehicles intend to access LBS without revealing their precise location information to the RSU due to privacy concerns. Since the RSU requires information on the location of vehicles to provide LBS, \cite{ManishTDSC} proposed a mechanism where vehicles share their coarse-grained location information only with the RSU, not with intermediate vehicles that forward the packet. In their method, each forwarding vehicle embeds its node identity and location information into the packet as it traverses to the RSU. As a result, their method incurred a high end-to-end delay due to processing overhead at each forwarding vehicle associated with embedding provenance information in the packet. Pointing out this drawback, the following section formally poses a new problem statement for designing privacy-preserving protocols that capture data- and spatial-provenance with high reliability while maintaining low end-to-end delay on the packets.\looseness=-1

\begin{figure}[ht!]
     \centering
    \includegraphics[trim={0 0 0 0},clip,scale=0.45]{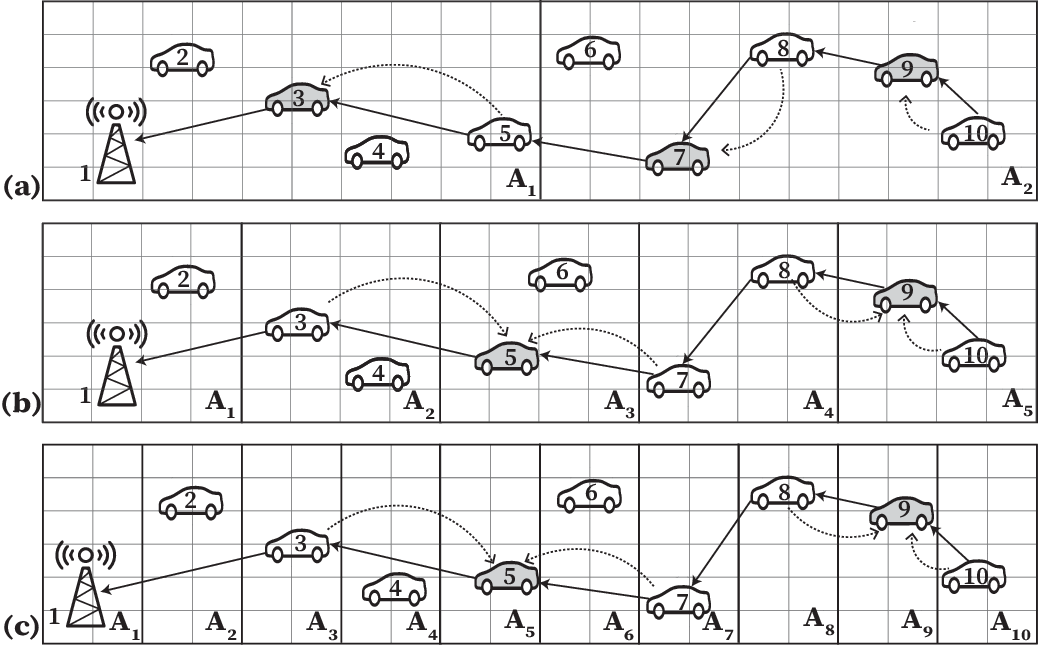}  
        \caption{A depiction of relaxed-privacy \textit{vis-à-vis} area segmentation. The background grid depicts the precision of GPS. Grey-colored vehicles implement our BSE and TSE schemes by embedding the information received from their neighboring (white-colored) vehicles.\looseness=-1 }
        \label{fig: network Image_car}
         \vspace{-0.45cm}
\end{figure}

\subsection{Problem Statement}
To balance the contrasting needs of the RSU and the vehicle's privacy concerns, \cite{ManishTDSC} proposed a multihop protocol where vehicles share their coarse-grained location information only with the RSU and not with intermediate vehicles that forward the packet. To achieve joint data- and spatial-provenance with such privacy constraints, each forwarding vehicle in the framework of \cite{ManishTDSC} embeds the identity of the corresponding location information in the packet along the path traveled, thereby incurring a high end-to-end delay. As a result, such methods are not applicable when there are stringent deadlines on the end-to-end delay of the packets. In the context of reducing end-to-end delay in data-provenance methods, \cite{amogh} suggests a skipping method, where some vehicles skip the provenance embedding process on the packets while the other vehicles cooperatively embed the provenance of their neighboring vehicles to reduce the end-to-end delay. While skipping-based methods to reduce end-to-end delay are available for recovering data-provenance, no such methods exist to provide benefits in end-to-end delay for joint data- and spatial-provenance. 
Inspired by this research gap, we ask whether skipping-based protocols can jointly recover data- and spatial-provenance under stringent delay, and what trade-offs they induce in reliability, privacy, overhead, and practical feasibility?

\begin{table*}
\centering
\caption{Comparison with relevant works, where joint provenance implies joint data- and spatial-provenance.}
\label{table:summary}

\resizebox{0.8\textwidth}{!}{%
\begin{tabular}{|c|c|c|c|c|c|c|}\hline
Reference
& Data-provenance
& Spatial-provenance
& Joint provenance
& Spatial-privacy
& Low-latency
& Testbed-results \\ \hline

\cite{Data_provenance_in_vehicle_data_chains, sultana_2011, sultana_2015,
Dictionary_based_provenance, IoV_data_plausibility, porkodi_secure_2021,
RPL_Based_provenance_IOT, provenance_compression_IoT, multi-hop_povenance_for_IoT}
& \checkmark & X & X & X & X & \checkmark \\ \hline

\cite{amogh, suraj}
& \checkmark & X & X & X & \checkmark & \checkmark \\ \hline

\cite{5G_multipath, mutipath_localization_1, V2V_assisted_coperative_localization, STAMP}
& X & \checkmark & X & \checkmark & X & X \\ \hline

\cite{ManishTDSC, manish_WCNC, manish_COMSNETS}
& \checkmark & \checkmark & \checkmark & \checkmark & X & \checkmark \\ \hline


Proposed work
& \checkmark & \checkmark & \checkmark & \checkmark & \checkmark & \checkmark \\ \hline
\end{tabular}
}

\vspace{-0.3cm}
\end{table*}

\subsection{Contributions}
To meet the RSU's localization requirements along with the privacy constraints of vehicles, we use the strategy suggested in \cite{ManishTDSC}, in which the RSU partitions its coverage area into uniform segments\footnote{\label{linear}This segmentation concept is relevant to complex road models \cite{IoV_models}, including bi-directional roads with numerous lanes and multi-layered roads, such as elevated roads and bridges. For the aforementioned generalized models, our concepts can be executed by assigning a dedicated RSU to each road or utilizing a single RSU to serve multiple roads through multiple-access techniques.}, subsequently directing vehicles to incorporate the identities (IDs) of their respective segments when transmitting packets.\footnote{This design is in line with regulatory standards such as IEEE~1609.2 \cite{WAVE} and ETSI Intelligent Transportation Systems (ITS) \cite{ETSI_TR_103_415}.} Due to non-disclosure of precise vehicle locations in this model, this privacy framework is referred as a \emph{relaxed-privacy} model. For illustration, in Fig. \ref{fig: network Image_car}, multiple scenarios are depicted in which the coverage area of the RSU is partitioned into distinct segments according to the privacy requirements of the vehicles. While we use the segmentation strategy of \cite{ManishTDSC}, we further relax the privacy model by requiring vehicles to explicitly share their segment information with neighboring vehicles. Fig. \ref{fig: network Image_car} captures our additional relaxation on the privacy requirement where vehicles share their coarse-level location information with their neighboring vehicles. Note that arrows in the figure indicate the direction of flow of segment IDs between vehicles. By exploiting the above mentioned relaxation in the privacy model, we make the following contributions, as outlined below:\looseness=-1

    1) To meet low-latency requirements, we propose two protocols, namely: Bi-Segment Embedding (BSE) and Tri-Segment Embedding (TSE). Both protocols are built on the principles of \emph{skipping} and \emph{relaxed-privacy} to reduce per-hop delays and use Bloom filters as a data structure to embed the data- and spatial-provenance information jointly. The terms \emph{bi-segment} and \emph{tri-segment} refer to the number of segment IDs a node embeds on the packet at a time in order to reduce end-to-end delays: two in BSE as shown in Fig. \ref{fig: network Image_car} (a), and three in TSE as depicted in Fig. \ref{fig: network Image_car} (b,c). For more details, refer to Sections \ref{sec:network_model} and \ref{sec:Bloom Filters}.\looseness=-1
    
    2) Although Bloom filters aid in low-latency recovery, they are accompanied by error rates owing to false positive events, which depend on the Bloom filter size and the number of Hash functions. To select Bloom filter parameters when employing BSE and TSE, we derive analytical expressions for error rates in retrieving spatial-provenance when the underlying wireless technology used in the vehicular network has an arbitrary transmission range (see Section \ref{sec:optimization} and Section \ref{sec:optimization_of_BF2}).\looseness=-1  
    
    3) We validate the evaluation of proposed analytical expressions by comparing them with experimental results, which reveal that the analytical expressions yield near-optimal Bloom filter configurations. After presenting a thorough analysis of the trade-off between privacy and error rates, we compare the performance of our proposed protocols with the baseline method in \cite{ManishTDSC} and show that we get significant benefits in reliability as well as end-to-end delay (see Section \ref{sec:results}).\looseness=-1 

    4) In BSE and TSE, vehicles explicitly share their coarse-grained location information with the neighboring vehicles. To incorporate practical aspects, we explore physical-layer-based localization methods to indirectly infer the segment IDs of neighboring vehicles. Experiments reveal that the proposed methods continue to provide benefits over the baseline protocols (see Section \ref{sec:Imperfect_neighbor}).\looseness=-1

    5) Finally, to understand the practical insights on implementing proposed protocols in realistic vehicular networks, we perform timing analysis using the radio characteristics of ZigBee and 5G stack, and as a takeaway point, we demonstrate that our proposed protocols offer significant benefits in latency as compared to existing techniques in \cite{manish_WCNC, manish_COMSNETS} when implemented over 5G stack. We also evaluate the performance of our proposed protocols under mobility and dynamic changing topology scenarios, and the results confirm their robustness in realistic V2X networks.\looseness=-1

\subsection{Related Works \& Novelty}
\begin{table}[!t]
\centering
\renewcommand{\arraystretch}{1.25}
\scriptsize
\setlength{\tabcolsep}{3pt}

\caption{Comparison of key metrics w.r.t. the baseline LNE \cite{ManishTDSC}.}
\label{tab:LNE_BSE_TSE_Comparison}
\resizebox{6.5 cm}{!}{
\begin{tabular}{|
p{3.5cm}|
>{\centering\arraybackslash}p{3.2cm}|
}
\hline
\textbf{Metric description} & \textbf{Relative ordering} \\ \hline

\textbf{Computation load at nodes} 
& $LNE$ \cite{ManishTDSC} $> BSE > TSE$ \\ \hline

\textbf{Computation load at RSU} 
& $LNE$ \cite{ManishTDSC} $< BSE < TSE$ \\ \hline

\textbf{Error rates} 
& $LNE$ \cite{ManishTDSC} $> BSE > TSE$ \\ \hline

\textbf{Location privacy w.r.t. vehicles} 
& $LNE$ \cite{ManishTDSC} $> BSE = TSE$ \\ \hline

\textbf{Location privacy w.r.t. RSU} 
& $LNE$ \cite{ManishTDSC} $= BSE = TSE$ \\ \hline

\textbf{End-to-end delay} 
& $LNE$ \cite{ManishTDSC} $> BSE > TSE$ \\ \hline

\end{tabular}}
\vspace{-0.4cm}
\end{table}

Contributions most relevant to our study on data-provenance are \cite{Data_provenance_in_vehicle_data_chains, sultana_2011, sultana_2015, Dictionary_based_provenance, IoV_data_plausibility, porkodi_secure_2021, RPL_Based_provenance_IOT, provenance_compression_IoT, multi-hop_povenance_for_IoT, amogh,suraj}, those on spatial-provenance are \cite{5G_multipath, mutipath_localization_1, V2V_assisted_coperative_localization, STAMP, LPPS_2}, and those on joint data- and spatial-provenance are \cite{ManishTDSC,manish_WCNC, manish_COMSNETS}. Data-provenance techniques suggested in \cite{Data_provenance_in_vehicle_data_chains, sultana_2011, sultana_2015, Dictionary_based_provenance, IoV_data_plausibility, porkodi_secure_2021, RPL_Based_provenance_IOT, provenance_compression_IoT, multi-hop_povenance_for_IoT} are not suitable for multihop networks, whereas\cite{amogh,suraj} do not capture the spatial location in multihop networks. Spatial-provenance techniques in \cite{5G_multipath, mutipath_localization_1, V2V_assisted_coperative_localization, STAMP} focus solely on learning the spatial location of the source node. The joint data- and spatial-provenance technique in \cite{ManishTDSC,manish_WCNC, manish_COMSNETS} incurs a high end-to-end delay while attempting to provide relaxed-privacy features with respect to the RSU and the other vehicles. Overall, while \cite{sultana_2015, amogh, suraj} and \cite{ManishTDSC, manish_WCNC, manish_COMSNETS} use Bloom filters for capturing provenance, none of them address a skipping-based joint data- and spatial-provenance technique, which reduces end-to-end delay without sacrificing traceability in privacy-constrained multihop networks. Comparative aspects of our work with existing techniques are summarized in Table \ref{table:summary}.

Since our work addresses joint data- and spatial-provenance, \cite{ManishTDSC} serves as a natural baseline for comparison. Henceforth, throughout this paper, we refer to this baseline as location node embedding (LNE), as provenance information is embedded by every forwarding node. Based on Table \ref{table:summary}, our proposed framework differs from LNE \cite{ManishTDSC} primarily in terms of low-latency, while sharing several high-level design aspects. To compare other common features in detail, Table \ref{tab:LNE_BSE_TSE_Comparison} provides a detailed comparison across computation overhead, error rates, privacy, and end-to-end delay. The results indicate improvements in computation cost at forwarding nodes, error rates, and end-to-end delay, while trading off privacy with respect to neighboring vehicles.\looseness=-1 

\section{Network Model} \label{sec:network_model}
\bl{We consider a localized coverage area of a single RSU, representing a small portion of a vehicular network in which multiple RSUs are interconnected via a secure backhaul. The localized coverage area contains $N$ nodes, comprising $N-1$ mobile nodes and a fixed RSU.}
Mobile nodes want to communicate with the RSU; however, due to limited power constraints, direct communication with the RSU is not possible.
Consequently, nodes communicate with the RSU through multiple intermediate nodes employing a decode-and-forward approach \cite{DecodeAndForward} in a multihop configuration. \bl{To support multihop communication through intermediate nodes, we outline the authentication method employed in our model as given in the next section.}\looseness=-1
\vspace{-0.3 cm}
\subsection{Network Authentication}
\bl{We assume that a vehicle entering the coverage area of the RSU authenticates itself with a gateway node using a pre-shared root key known to the core network. Upon successful authentication, each vehicle derives a unique secret key $K_y,~ y \in [N-1]$ using standard key derivation functions. The RSU maintains the credentials of all the authenticated vehicles within its coverage area. However, the individual secret keys $K_y$  remain private to each vehicle. Subsequently, each vehicle performs neighbor discovery and authenticates with neighboring nodes using a public-key cryptosystem. Only vehicles that have completed this authentication procedure and possess a valid $K_y$ are allowed to participate in multihop communication. For more details on authentication and neighbor discovery mechanisms, the reader may refer to \cite[Section~2]{amogh} and \cite[Section~2.2]{suraj}.\looseness=-1}

To enhance network security and provide advanced network diagnostics, the RSU wants to determine the path traveled by the packets and the precise GPS locations of all network nodes that have relayed them. However, considering privacy concerns, nodes may refrain from disclosing their precise GPS coordinates to either the RSU  or other registered nodes in the network \cite[Section~1]{LPPS}. \bl{To cater to these requirements, the following section formally outlines the privacy requirements and threat model adopted in this work.\looseness=-1}
\vspace{-0.5 cm}
\subsection{Privacy and Threat Model}
\bl{As described above, our network contains $N-1$ authenticated mobile nodes along with the RSU, and may also include other non-authenticated nodes within the coverage area. In our threat model, all entities are assumed to be honest-but-curious, passively observing network traffic to infer vehicle locations. All entities follow the protocol faithfully but may attempt to infer the exact GPS coordinates of other vehicles from observed packets. We assume that no entity modifies packets or launches active attacks. In particular, these entities are classified into three types based on their roles and observation capabilities: forwarding nodes, non-forwarding nodes, and the RSU. Forwarding nodes are authenticated vehicles participating in the packet forwarding. They have access to relayed packets and embedded Bloom filters. Non-forwarding nodes are passive observers with no role in packet delivery and are further divided into two subtypes. Authenticated non-forwarding nodes are registered vehicles within the RSU coverage area. They are not participating in the ongoing packet-forwarding process; however, they have full protocol knowledge and can observe network traffic. Non-authenticated nodes have not registered with the RSU and therefore fall outside the system's trust boundary. They can overhear wireless packets within their communication range; however, they lack knowledge of the protocols. The RSU is a trusted infrastructure node that collects all transmitted packets to recover provenance information, potentially inferring exact vehicle locations.\looseness=-1}
\vspace{-0.3 cm}
\subsection{Approach for Relaxed-Privacy}
\bl{As discussed above, vehicles are unwilling to reveal their exact GPS coordinates; however, they are willing to share coarse-level location information to avail the benefits of LBS.} To reconcile the demands of the RSU with the location privacy of the nodes, we take the approach mentioned in \cite{ManishTDSC, manish_WCNC}, where the RSU partitions its coverage area into $r$ geographical areas, where $r \in \mathbb{N}$, and instructs the nodes to disclose the identities of their region while relaying the packet. In this context, $r$ is an arbitrary quantity, which is \emph{a priori} determined through a conciliatory agreement between the privacy needs of the nodes and the precision demands of the RSU.\looseness=-1

In the context of vehicular networks, we assume that the RSU defines its coverage area along a linear path, partitioning it into smaller parts of uniform regions, henceforth termed as segments. The collection of segments is represented as $\Delta \triangleq \{A_1,A_2,\ldots,A_{r}\}$, as seen in Fig. \ref{fig: network Image_car} for $r =2,5$ and $10$. We assume that the RSU is located in segment $A_1$, specifically at one corner of the coverage area, whereas the remaining mobile nodes are scattered randomly throughout $\Delta$. We also assume that the segments farther away from the RSU are allotted higher index numbers.
The farthest segment from the RSU inside the coverage area is $A_{r}$. The segmentation of the coverage area of the RSU balances the relaxed-privacy requirements of the nodes with the RSU's need for learning the location of the nodes. We assume that the RSU has high transmit power compared to other nodes in the network; therefore, its broadcast messages can propagate to the farthest segment, namely $A_r$, on the downlink. To aid the nodes in learning the segment ID, the RSU disseminates a dictionary that delineates the GPS boundaries corresponding to various segments. As the nodes possess GPS devices, they can independently ascertain their segment ID utilizing the disseminated dictionary. For more information about privately sharing the dictionary via a broadcast message, we refer the reader to \cite{spatial}.\looseness=-1  

\begin{figure}
    \centering
    \includegraphics[scale=1]{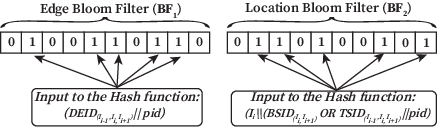}
    \caption{The diagram illustrates the embedding process in CLBF at the intermediate nodes, which consist of edge and location Bloom filters.\looseness=-1}
    \label{fig:CLBF}
    \vspace{-0.5cm}
\end{figure}

To explain the spatial-provenance embedding process, the identity of the source node is marked as $I_s$, the identity of the RSU is marked as $I_{N}$, and the identities of the other mobile nodes are represented as $I_n$, where $n \in [1,~N-2]$. Furthermore, to establish the correspondence between $\Delta$ and the GPS coordinates of the nodes, we utilize the function $g: \{I_1, I_2,\ldots,I_{N-2}, I_s\} \rightarrow \Delta$ to convert node identities into segment identities based on their current locations. \bl{We assume that the function $g$, which defines the segment mapping based on the privacy preferences agreed upon by the participating vehicles, remains constant over a sufficiently long duration. This assumption is motivated by the fact that vehicles typically do not alter their privacy preferences at millisecond time scales during ongoing communication.\footnote{\bl{The frequency of dictionary broadcasts, possible updates to $g$, and their impact on system performance are discussed in Section \ref{sec:mobility analysis}.}}}
Nonetheless, the function $g$ may alter at a subsequent coherence time when another packet is transmitted. The following section highlights the constraints on the routing protocols used by nodes for uplink communication with the RSU.\looseness=-1

\subsection{Network Constraints} \label{sec:Routing Constraints}
When packets are sent from the source node $I_{s}$ to the RSU, we assume that the network complies with the following routing and communication constraints.\looseness=-1

1) Routing Constraint: A routing protocol, such as Ad-hoc On-Demand Distance Vector (AODV) \cite{aodv_original}, is utilized to facilitate minimum-hop packet delivery from the source node to the RSU. As a result, the routing protocol prevents the formation of any loops or back-hops. Formally, the packet flow ensures that a node in segment $A_i$, for $i \in [r]$, does not pass a packet to a node in segment $A_j$, where $j \in [r]$ and $i < j$.\looseness=-1

2) Communication Constraints: We assume that nodes can communicate only within a specific distance, and all nodes have a uniform transmission range. We utilize the parameter $\beta$, where $\beta \in \mathbb{N}$, to represent the maximum number of contiguous segments a mobile node can communicate with. For instance,  a node is in segment $A_4$ and $\beta=3$, it can transmit packets only to nodes inside segments $A_4$, $A_3$, $A_2$, or $A_1$, and correspondingly, it can directly receive packets from nodes in segments $A_4$, $A_5$, $A_6$, or $A_7$. We represent this communication constraint mathematically using the function $S: \Delta \rightarrow [r]$, which maps the segment identity $A_i$ into its numerical identity $i$, for $i \in [r]$. For any two consecutive nodes with identities $I_i$ and $I_j$ in a $h$-hop path; then they must adhere to the communication constraint $ |S(g(I_i))-S(g(I_j))|\leq \beta$, where $i\neq j$ and $i, j \in [N]\backslash \{N-1\} \cup \{s\}$.\looseness=-1  

3) Relaxed-Privacy Constraints: We assume that the nodes have relaxed-privacy constraints with respect to their neighbors, i.e., nodes can explicitly share their segment IDs with their neighbors in contrast to the network model in \cite{ManishTDSC, manish_WCNC}, where the nodes do not share their segment IDs with their neighbors.\looseness=-1

4) Delay Constraints: To support low-latency constraints on the vehicular network, the underlying protocol desires to achieve low end-to-end delay while the packet traverses from $I_S$ to the RSU. Although the method in \cite{ManishTDSC, manish_WCNC} also achieves low end-to-end delay, we assume that delay constraints are more stringent in our network model.\looseness=-1

Given the aforementioned network constraints, the nodes aim to encode information regarding their segment IDs and the path traversed by the packet, enabling spatial-provenance recovery at the RSU upon packet reception. 
In the next section, we introduce a Bloom filter \cite{Bloom_filter_1} based data structure to accomplish the aforementioned objectives.\looseness=-1

\section{Bloom Filter-based Spatial-Provenance Recovery} \label{sec:Bloom Filters}
To recover the path traversed by the packet and the segment IDs of the nodes that have forwarded the packet, we employ Correlated Linear Bloom Filters (CLBF) using two Bloom filters, as illustrated in Fig. \ref{fig:CLBF}, where one Bloom filter encodes the packet path, and the other encodes segment IDs of forwarding nodes. The decision to employ Bloom filters is motivated by their fixed size, no false negatives, and constant time complexity for insertion and query operations.\looseness=-1


To support low-latency operation, we use the deterministic-double-edge (DDE) embedding method \cite{amogh, suraj} to embed the path traversed by the packet. For localization, we use bi-segment embedding (BSE) and tri-segment embedding (TSE), which skip embedding at selected intermediate nodes. Since both DDE and location embedding employ skipping, they collectively reduce end-to-end delay.\looseness=-1
\begin{figure}[!]
     \centering
    \includegraphics[trim={0cm 0cm 0cm 0cm},clip,scale=0.23]{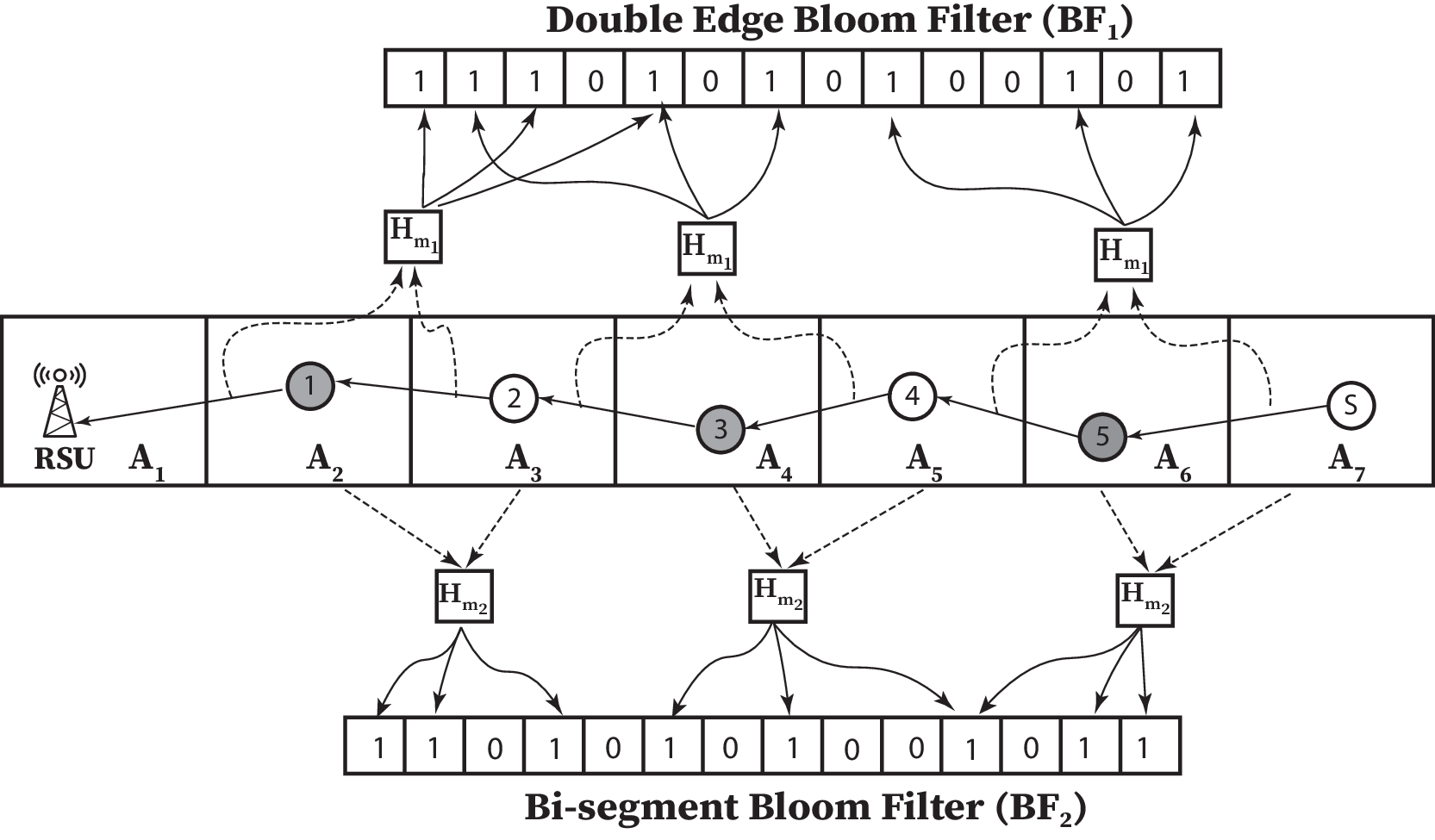}  
        \caption{A depiction of the embedding procedure for BSE.\looseness=-1}
        \label{fig:LSE_embedding}
         \vspace{-0.5cm}
\end{figure}

To aid the provenance embedding process, the packet structure includes a CLBF and a counter to convey hop-length information to the RSU. The CLBF has two Bloom filters, denoted as $\textbf{BF}_{1} \in \{0, 1\}^{m_{1}}$ and $\textbf{BF}_{2} \in \{0, 1\}^{m_{2}}$, wherein $\textbf{BF}_{1}$ embeds the path traveled by the packet, and $\textbf{BF}_{2}$ embeds the segment IDs of the nodes. At a high-level, when a node relays the packet, it can perform one or more of the following tasks: (i) embedding onto $\textbf{BF}_1$ the identity of the DDE that represents the connection between three nodes- the packet forwarding node, the current embedding node, and the next-hop node, using $k_{1}$ number of Hash functions for $1\leq k_1\leq m_1$ \cite{amogh, suraj}, (ii) embedding the segment ID (explained in subsequent section) onto $\textbf{BF}_2$ using $k_{2}$ number of Hash functions for $1\leq k_2 \leq m_2$, and (iii) incrementing the hop-counter. Upon receiving the packet, the RSU can retrieve both the path traveled by the packet and the segment IDs of the forwarding nodes. Next, we describe how path information is embedded onto $\textbf{BF}_1$.\looseness=-1
\vspace{-0.3cm}
\subsection{Ingredients for Embedding Path Information} \label{sec:BF_ingredients}
From the set of $N$ nodes in the network, a packet travels through a set of $h$ nodes, where $h \leq N$, before reaching the RSU. Assume in a $h$-hop path, nodes $I_x$, $I_y$, and $I_z$, where $x,y,z \in [N] \backslash \{N-1\} \cup \{s\}$, are directly connected neighbors within transmission range and the packet is forwarded from node $I_z$ to node $I_y$ and from $I_y$ to node $I_x$, forming $I_z \rightarrow I_y \rightarrow I_x$ packet flow as per the routing protocol defined in Section \ref{sec:Routing Constraints}. For the DDE embedding process onto $\textbf{BF}_{1}$, we assume that each node $I_y$, where $y \in [N] \backslash \{N-1\} \cup \{s\}$, is equipped with a pre-shared key with the RSU, denoted as $K_y$. As part of the packet forwarding process through a $h$-hop path, node $I_y$ dynamically generates a distinct double-edge ID denoted as $DEID_{(x,y,z)}$ using the node ID of its neighbors, i.e., $I_{x}$ and $I_{z}$. One such approach to produce $DEID_{(x,y,z)}$ involves utilizing $I_{x}$, $I_{y}$ and $I_{z}$ in conjunction with the key $K_y$ to derive a string $DEID_{(x, y,z)} = (I_{x}||I_y||I_{z} ||K_{y})$, where $||$ represents a string concatenation operation. Note that the function $DEID_{(\cdot, \cdot,\cdot)}$, is not symmetric, i.e., $DEID_{(x,y,z)}$ is not equal to $DEID_{(z,y,x)}$, where $DEID_{(z,y,x)} = (I_{z} || I_{y}|| I_{x} || K_{y})$ is the double-edge ID derived at node $I_y$ for the double-edge link from node $I_{x}$ to node $I_y$ and from node $I_{y}$ to node $I_{z}$. This double-edge ID creation happens on the fly during the packet transmission along the $h$-hop path. We assume that the RSU has access to the IDs and secret keys of all the nodes, allowing it to obtain a set of all double-edge IDs in the network, represented as $\mathcal{DE} = \{DEID_{(x,y,z)}~|~ \forall ~x,y,z \in [N] \backslash \{N-1\} \cup \{s\}; ~x \neq y,~ y\neq z,~x\neq z \}$. Next, we will explain the ingredients for embedding segment IDs onto $\textbf{BF}_2$.\looseness=-1
\vspace{-0.3cm}
\subsection{Ingredients for Embedding Segment IDs} \label{sec:BF2_ingredients}
We assume that every node in the network knows its segment ID using the dictionary that the RSU broadcasts. To support low-latency, we assume that some intermediate nodes skip embedding their segment IDs on $\textbf{BF}_2$ to reduce end-to-end delay. The nodes that skip the embedding of their segment ID directly on $\textbf{BF}_2$ are termed forwarding nodes, and the nodes that embed their segment ID on $\textbf{BF}_2$ are termed embedding nodes. The forwarding nodes share their segment IDs explicitly with the embedding nodes.
Subsequently, the embedding node will jointly perform the insertion operation of its own segment ID and the segment IDs of its neighbors on $\textbf{BF}_2$. BSE and TSE methods are also termed perfect neighbor location embedding methods because neighbors explicitly share their segment ID with an embedding node. Next, we will proceed towards the ingredients for embedding segment IDs in BSE and TSE.\looseness=-1

\subsubsection{Ingredients for Bi-Segment Embedding (BSE)}\label{sec:Ingredients_BSE}
For bi-segment embedding onto $\textbf{BF}_2$, we assume that each node $I_y$ which has a neighbor node $I_{z}$, for $I_y \leftarrow I_z$ link in a $h$-hop path, where $y,z \in [N]\backslash\{N-1\}\cup \{s\}$, creates a distinct bi-segment ID denoted as $BSID_{(y,z)}$. $BSID_{(y,z)}$ is created by the node $I_y$ dynamically on the fly as the packet is received by $I_y$. One such approach to produce $BSID_{(y,z)}$ involves utilizing the $g(I_y)$, $g(I_{z})$ in conjunction with the key $K_y$ to derive a string $BSID_{(y,z)} = (g(I_y)||g(I_{z})|| K_y)$. Note that the function $BSID_{(\cdot,\cdot)}$ is not symmetric, i.e., $BSID_{(y,z)} \neq BSID_{(z,y)}$, where $BSID_{(z,y)} = (g(I_{z})||g(I_y)|| K_z)$ is the  bi-segment ID derived at the node $I_{z}$ for the bi-segment link from node $I_y$ to the node $I_{z}$. 
We assume that the RSU has access to the segment IDs and secret keys of all the nodes, allowing it to obtain the set of all BSIDs in the network, represented as $\mathcal{BS} = \{BSID_{(y, z)}|~\forall y,z \in [N]\backslash\{N-1\}\cup \{s\}; y\neq z\}$.\looseness=-1

\begin{figure}[!]
     \centering
    \includegraphics[trim={0cm 0cm 0cm 0cm},clip,scale=0.23]{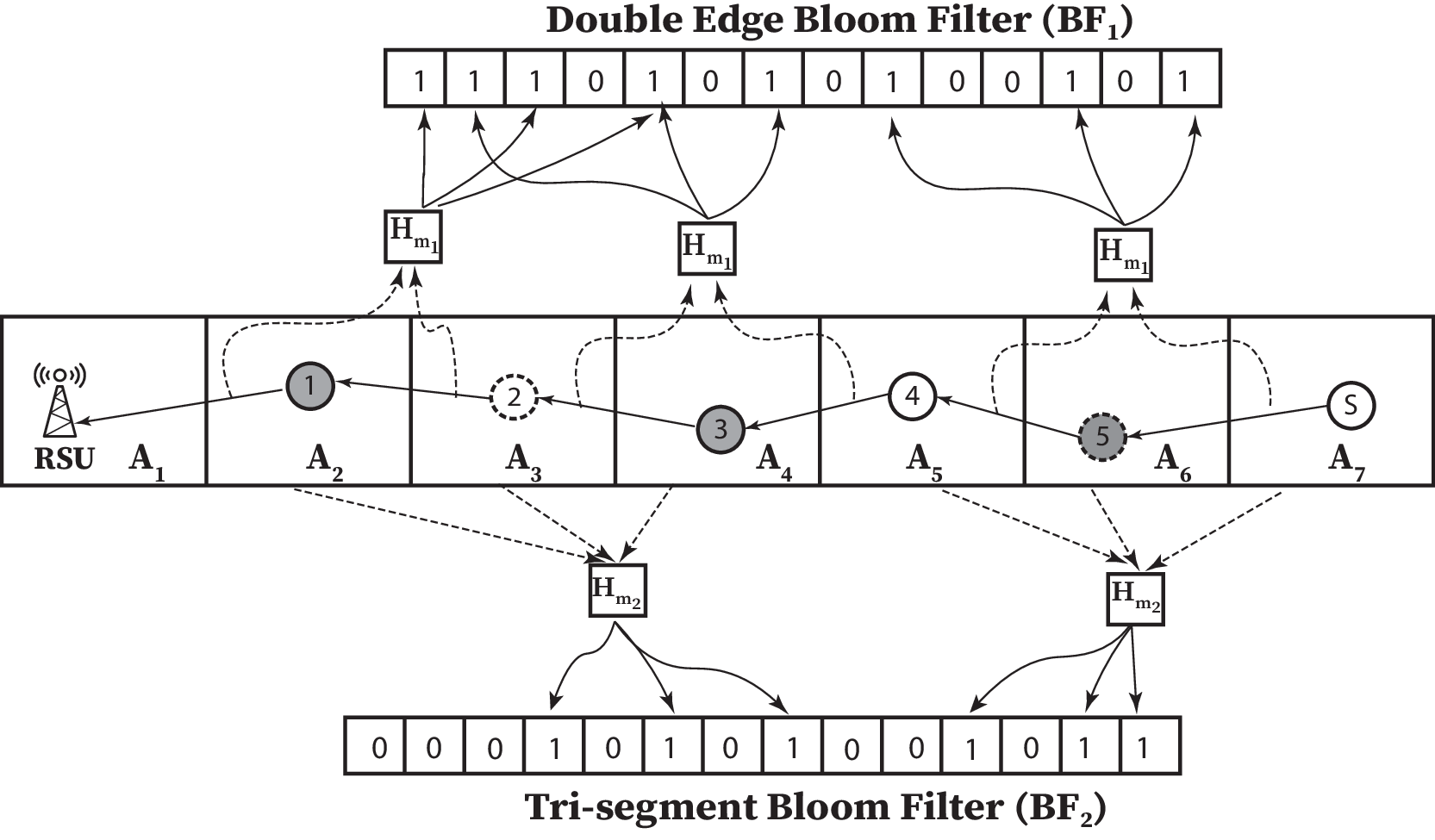}  
        \caption{A depiction of the embedding procedure for TSE.\looseness=-1}
        \label{fig:LDE_embedding}
         \vspace{-0.3cm}
\end{figure}
\subsubsection{Ingredients for Tri-Segment Embedding (TSE)}
\label{sec:Ingredients_TSE}
For tri-segment embedding onto $\textbf{BF}_2$, we assume that each node $I_y$ which has two neighbors $I_{x}$ and $I_{z}$ for $I_z \rightarrow I_y\rightarrow I_x$ link in a $h$-hop path, where $x,y,z \in [N]\backslash\{N-1\} \cup \{s\}$, creates a distinct tri-segment ID denoted as $TSID_{(x,y,z)}$. $TSID_{(x,y,z)}$ is created on the fly at the node $I_{y}$ during packet transmission when node $I_y$ receives the packet from $I_z$ and before forwarding to node $I_x$.\footnote{On packet reception from node $I_z$, node $I_{y}$ may ask the node $I_{x}$ to share its segment ID with node $I_{y}$. The effect on end-to-end delays is explained in Section \ref{sec:results}.} One such approach to produce $TSID_{(x,y,z)}$ involves utilizing the $g(I_{x})$, $g(I_y)$ and $g(I_{z})$ in conjunction with the key $K_y$ to derive a string $TSID_{(x,y,z)} = (g(I_{x})||g(I_y)||g(I_{z})||K_y)$. Note that the function $TSID_{(\cdot, \cdot,\cdot)}$ is not symmetric, i.e., $TSID_{(x,y,z)} \neq TSID_{(z,y,x)}$, where $TSID_{(y,z,x)} = (g(I_{z})||g(I_y)||g(I_{x})||K_y)$ is the tri-segment ID derived at the node $I_{y}$ for the tri-segment link from node $I_{z}$ to the node $I_{y}$ and from $I_{y}$ to node $I_{x}$. We assume that the RSU has access to the segment IDs and secret keys of all the nodes, allowing it to obtain the set of all TSIDs in the network, represented as $\mathcal{TS} = \{TSID_{(I_{x},I_{y},I_{z})}|\forall x,y,z \in [N]\backslash\{N-1\} \cup \{s\}; x\neq y, y\neq z, x\neq z \}$.\looseness=-1
\vspace{-0.3cm}
\subsection{Embedding Process for Spatial-Provenance} 
\label{sec:BF_embedding}
To explain the provenance embedding process, we assume that the packets are forwarded from the source node $I_{s}$ through a $h$-hop path $I_{s} \rightarrow I_{j_{1}} \rightarrow I_{j_{2}} \rightarrow \ldots \rightarrow I_{j_{h-2}} \rightarrow I_{j_{h-1}}$, where $j \in [N]\backslash\{N-1\}$, before reaching the RSU. To start with, both Bloom filters are initialized to all-zero vectors, and the counter is also initialized to zero. Now, we will explain the embedding process of BSE followed by TSE. 

\subsubsection{BSE Process for Spatial-Provenance}\label{sec:BSE_Embedding}
In BSE process, the source node $I_s$ leaves both $\textbf{BF}_1$ and $\textbf{BF}_2$ untouched, increments the hop-counter to one, and transmits the packet to the next node $I_{j_1}$ in the $h$-hop path. 

Upon receiving the packet from node $I_s$, node $I_{j_{1}}$ updates both $\textbf{BF}_{1}$ and $\textbf{BF}_{2}$ as follows. To embed the double-edge ID connecting node $I_{s}$, $I_{j_{1}}$ and $I_{j_{2}}$ onto $\textbf{BF}_{1}$, a Hash function $H_{m_1}(\cdot)$ that randomly picks an index in $[1, m_{1}]$ with uniform distribution is used. Then, node $I_{j_{1}}$ generates $k_{1}$ indices denoted by $\{w_{I_{j_{2}},I_{j_{1}}, I_s }^{(a)} \in [1, m_{1}]| ~1\leq a \leq k_{1}\}$ in an independent manner as
\begin{equation}
\label{eq:embed_BF1}
w_{I_{j_{2}},I_{j_{1}}, I_s }^{(a)} = H_{m_1}(DEID_{(I_{j_{2}},I_{j_{1}},I_{s})}||pid||a),
\end{equation}
where $DEID_{(I_{j_{2}},I_{j_{1}},I_{s})}$ denotes the double-edge ID for the link $I_{s} \rightarrow I_{j_{1}} \rightarrow I_{j_{2}}$ generated at $I_{j_{1}}$, and $pid$ signifies the packet ID. Subsequently, the contents of $\textbf{BF}_{1}$ is set to one as $\textbf{BF}_{1}[w_{I_{j_{2}},I_{j_{1}}, I_s }^{(a)}] = 1$ on the chosen set of indices $\{w_{I_{j_{2}},I_{j_{1}}, I_s }^{(a)}  \in [1,m_1]|~ 1\leq a \leq k_{1}\}$. After embedding $\textbf{BF}_{1}$, node $I_{j_{1}}$ proceeds to embed its BSID onto $\textbf{BF}_{2}$.\looseness=-1

Using a Hash function $H_{m_2}(\cdot)$ that randomly picks an index in $[1, m_{2}]$ with uniform distribution, node $I_{j_1}$ generates $k_{2}$ indices denoted by $\{v_{I_{j_1}, I_s}^{(l)} \in [1, m_{2}]| 1\leq l \leq k_{2}\}$ in an independent manner as\looseness=-1
\begin{equation}
\label{eq:embed_BF2_LEID}
v_{I_{j_1}, I_s}^{(l)} = H_{m_2}(I_{j_1}||BSID_{(I_{j_1}, I_s)}||pid||l),
\end{equation}
where $pid$ denotes the packet ID, and all the other notations are as explained earlier. Subsequently, the contents of $\textbf{BF}_{2}$ is set to one as $\textbf{BF}_{2}[v_{I_{j_1}, I_s}^{(l)}] = 1$ on the chosen set of indices $\{v_{I_{j_1},I_s}^{(l)} | 1\leq l \leq k_{2}\}$. Once $\textbf{BF}_2$ is updated at node $I_{j_1}$, the hop-counter is incremented by one, and the packet is forwarded to the next node $I_{j_{2}}$.\looseness=-1 

Along similar lines, the hop-counter is updated at every node, while $\textbf{BF}_{1}$,  $\textbf{BF}_{2}$ are updated at every alternate node when the packet is en route to the RSU through the rest of the  $h$-hop path, i.e., $I_{j_{2}} \rightarrow I_{j_{3}} \rightarrow \ldots \rightarrow I_{j_{h-2}} \rightarrow I_{j_{h-1}}$.\looseness=-1

\subsubsection{TSE Process for Spatial-Provenance}\label{sec:TSE_Embedding}
In TSE process, the source node $I_s$ leaves both $\textbf{BF}_1$ and $\textbf{BF}_2$ untouched, increments the hop-counter to one, and transmits the packet to the next node $I_{j_1}$ in the $h$-hop path.

Upon receiving the packet from node $I_s$, node $I_{j_{1}}$ updates both $\textbf{BF}_{1}$ and $\textbf{BF}_{2}$ as follows. To embed the double-edge ID connecting node $I_{s}$, $I_{j_{1}}$ and $I_{j_{2}}$ onto $\textbf{BF}_{1}$, node $I_{j_{1}}$ follows the same procedure for embedding double-edge as described for BSE. Subsequently, after embedding $\textbf{BF}_{1}$, node $I_{j_{1}}$ proceeds to embed its TSID onto $\textbf{BF}_{2}$.\looseness=-1

Using a Hash function $H_{m_2}(\cdot)$ that randomly picks an index in $[1, m_{2}]$ with uniform distribution, node $I_{j_1}$ generates $k_{2}$ indices denoted by $\{u_{I_{j_{2}},I_{j_{1}}, I_s }^{(t)} \in [1, m_{2}]| 1\leq t \leq k_{2}\}$ in an independent manner as\looseness=-1
\begin{equation}
\label{eq:embed_BF2_LDEID}
u_{I_{j_{2}},I_{j_{1}}, I_s }^{(t)} = H_{m_2}(I_{j_1}||TSID_{(I_{j_2},I_{j_1}, I_s)}||pid||t),
\end{equation}
where $TSID_{(I_{j_{2}}, I_{j_{1}}, I_{s})}$ denotes the tri-segment ID for the link $g(I_{s}) \rightarrow g(I_{j_{1}}) \rightarrow g(I_{j_{2}})$ generated at $I_{j_{1}}$ and all the other notations are as explained earlier. Subsequently, the bits in $\textbf{BF}_{2}$ are set to one as $\textbf{BF}_{2}[u_{I_{j_{2}},I_{j_{1}}, I_s }^{(t)}] = 1$ on the chosen set of indices $\{u_{I_{j_1}}^{(t)} | 1\leq t \leq k_{2}\}$. Once $\textbf{BF}_2$ is updated at node $I_{j_1}$, the hop-counter is incremented by one, and the packet is forwarded to the next node $I_{j_{2}}$.\looseness=-1 

Similarly, the hop-counter is updated at every node to reflect the number of hops encountered during its journey. $\textbf{BF}_{1}$ is updated at every alternate node, starting from $I_{j_1}$, such that updates occur at $I_{j_1},I_{j_3}, I_{j_5},\ldots,$ in a $h$-hop path until the destination node, i.e., the RSU is reached. On the other hand, $\textbf{BF}_{2}$ is updated at every third node, beginning with $I_{j_1}$, followed by updates at $I_{j_4},I_{j_7},\ldots,$, while the packet reaches the RSU through the rest of the path $I_{j_{2}} \rightarrow I_{j_{3}} \rightarrow \ldots \rightarrow I_{j_{h-2}} \rightarrow I_{j_{h-1}}$. Next, we explain the process of spatial-provenance recovery at the RSU.\looseness=-1

\subsection{Recovery Process for Spatial-Provenance} \label{sec:recovery}
When the packet reaches the RSU, the RSU recovers the path traveled by the packet and the corresponding segment IDs of the forwarding nodes. The recovery procedure has two stages: initially, the RSU retrieves the path traversed by the packet from the received $\textbf{BF}_1$, followed by the extraction of segment IDs from the received $\textbf{BF}_2$.\looseness=-1 

Utilizing $\textbf{BF}_{1}$, the RSU evaluates all prospective double-edges in the network, represented by $\mathcal{DE}$, to compile a set of recovered double-edges represented as $\hat{\mathcal{DE}}$, where $\hat{\mathcal{DE}} \subseteq \mathcal{DE}$. For retrieving the double-edges from $\textbf{BF}_1$, the RSU has all the necessary components, including node IDs $\{I_y\}$, packet ID $pid$, keys $\{K_y\}$, and double-edge IDs $\{DEID_{(x,y,z)} ~|~x,y,z \in [N]\backslash\{N-1\} \cup \{s\}\}$. In this context, a recovered double-edge from node $I_{z}$ to node $I_{y}$ and from node $I_{y}$ to node $I_{x}$ fulfills the requirement of $\textbf{BF}_{1}[p_{x,y,z}^{(a)}] = 1 $ for all $k_1$ indices, defined by $\{p_{x,y,z}^{(a)} \in [1,m_1]~|~ 1\leq a \leq k_1\}$, which are subsequently created at the RSU utilizing the double-edge IDs and all other credentials in a manner analogous to \eqref{eq:embed_BF1}. After that, the RSU employs a depth-first search algorithm on the retrieved set of double-edges $\hat{\mathcal{DE}}$ to derive a candidate set of $h$-hop paths from the source node $I_s$ to the RSU, denoted as $\mathcal{C} \triangleq \{\mathbf{c}_1, \mathbf{c}_2,\ldots,\mathbf{c}_{|\mathcal{C}|}\}$, where each $\mathbf{c}_i$ represents a $h$-hop path and $|\mathcal{C}| \in \mathbb{N}$.\looseness=-1 

From each candidate $h$-hop path $\mathbf{c}_i$ consisting of double-edges $\mathbf{c}_i=\{(I_s, I_{i_1},I_{i_2}), (I_{i_2},I_{i_3},I_{i_4} )\ldots (I_{i_{h-2}},I_{i_{h-1}}, I_{N})\}$, the RSU extracts the set of unique node IDs, which are represented as $\mathcal{V}_i = \{I_{s}, I_{i_{1}},I_{i_{2}}, \ldots, I_{i_{h-1}} \}$.
The set of unique node ID $\mathcal{V}_i$ is the ordered set of node IDs, where the cardinality of the set ${V}_i$ is $h$.\looseness=-1 

\subsubsection{Process of BSE Recovery} \label{sec:BSE_recovery}
For the recovery of segment IDs in BSE, from the set $\mathcal{V}_i$, starting from node $I_{i_1}$, the RSU extracts every alternate node in the set $\mathcal{V}'_i = \{I_{i_1},I_{i_3},I_{i_5},\ldots\}$.
Using the set of unique node IDs in $\mathcal{V}'_i$, the RSU generates the potential BSIDs to check their presence in $\textbf{BF}_2$. For creating the potential BSIDs, for a given node ID $v\in \mathcal{V}'_i$, RSU needs to have the segment ID of the node ID $v$, i.e., $g(v)$, and the segment ID of the node that has forwarded the packet to the node $v$. For the subscript of the segment ID of the particular embedding node $v \in \mathcal{V}'_1$, RSU has $r$ options, or it can be said that the segment ID of node $v$ will have a unique value from the set $\Delta$. The subscript of the segment ID of the forwarding node corresponding to the embedding node $v \in \mathcal{V}'_i$ lies in  $[S(g(v), \min(r, S(g(v))+\beta)]$, because the forwarding node may only lie in the segment which is within the communication range $\beta$ of node $v$ as defined in Section \ref{sec:Routing Constraints}. 
After this, for retrieving the bi-segment information from $\textbf{BF}_2$, the RSU has all the necessary components, including the node IDs of the embedding nodes $\{\mathcal{V}_i'\}$, packet ID $pid$, and the subscripts of the segment IDs to create the corresponding (nodeID, BSID) pairs as explained earlier.
Subsequently, it verifies the existence of each such pair in $\textbf{BF}_2$ by reconstructing the set of $k_{2}$ indices, analogous to \eqref{eq:embed_BF2_LEID}. In this context, a pair (nodeID, BSID) is referred to as a recovered pair if the values in $\textbf{BF}_{2}$ on all the chosen $k_{2}$ locations are set to one. Finally, for a given path $\mathbf{c}_{i}$, the RSU retrieves a set of possible segment IDs for each node, which is denoted as $\Delta_{i} = \{\Delta_{I_s}, \Delta_{I_{i_{1}}}, \Delta_{I_{i_{2}}}, \ldots, \Delta_{I_{i_{h-1}}}\}$. Note that all paths in $\mathcal{C}$ may not have valid bi-segments in $\textbf{BF}_2$, which adhere to the communication constraints defined in Section \ref{sec:Routing Constraints}.\looseness=-1 

\subsubsection{Process of TSE Recovery} \label{sec:TSE_recovery}
For the recovery of segment IDs in TSE, from the set $\mathcal{V}_i$, starting from node $I_{i_1}$, the RSU extracts the node ID of every third node resulting in a set $\mathcal{V}_i'' = \{I_{i_1},I_{i_4},I_{i_7},\ldots\}$. 
Using the set of unique node IDs in $\mathcal{V}_i''$, the RSU generates  TSIDs, where for each node ID in $\mathcal{V}_i''$, the RSU evaluates the $\Delta$ options to determine the possible segment ID of the node. 
Specifically, for each node $v \in \mathcal{V}_i''$, the RSU iteratively checks the set $\Delta$ to identify the valid TSIDs. 
After this, the RSU has all the necessary components, including the node IDs of the embedding nodes $\{\mathcal{V}_i''\}$, packet ID $pid$, and the subscripts of the segment IDs to create the corresponding (nodeID, TSID) pair as explained above. Subsequently, it verifies the existence of each such pair in $\textbf{BF}_2$ by reconstructing the set of $k_{2}$ indices, analogous to \eqref{eq:embed_BF2_LDEID}. In this context, a pair (nodeID, TSID) is referred to as a recovered pair if the values in $\textbf{BF}_{2}$ on all the chosen $k_{2}$ locations are set to one. At last, for the designated path $\mathbf{c}_{i}$, a set with possible segment IDs, which constitutes a subset of $\Delta$, is retrieved for each node, denoted as $\Delta_{i} = \{\Delta_{I_s}, \Delta_{I_{i_{1}}}, \Delta_{I_{i_{2}}}, \ldots, \Delta_{I_{i_{h-1}}}\}$. Note that all paths in $\mathcal{C}$ may not have valid tri-segments in $\textbf{BF}_2$, which adhere to the communication constraints defined in Section \ref{sec:Routing Constraints}.\looseness=-1 

In summary, using $\textbf{BF}_{1}$ and $\textbf{BF}_{2}$, the RSU produces the set $\{(\mathbf{c_{i}}, \Delta_{i})\} ~|~ i = 1, 2, \ldots, |\mathcal{C}|]$ with respect to the received packet ID $pid$. With the knowledge of the communication and routing constraints listed in Section \ref{sec:Routing Constraints}, RSU further filters the list of candidate paths and their associated segment IDs from the set $\{(\mathbf{c_{i}}, \Delta_{i})\} ~|~ i = 1, 2, \ldots, |\mathcal{C}|]$ to obtain $\{(\mathbf{c_{i}}, \bar{\Delta}_{i})\} ~|~ i = 1, 2, \ldots, |\mathcal{C}|]$, where $\bar{\Delta}_{i} = \{\bar{\Delta}_{I_s}, \bar{\Delta}_{I_{i_{1}}}, \bar{\Delta}_{I_{i_{2}}}, \ldots, \bar{\Delta}_{I_{i_{h-1}}}\}$ such that $\bar{\Delta}_{I} \subseteq \Delta_{I} ~\forall \, I \in \{I_s, I_{i_1}, \ldots, I_{i_{h-1}}\}$ are obtained after applying the constraints in Section \ref{sec:Routing Constraints}.\looseness=-1

Since the packet can traverse only a single path, the RSU should ideally recover a single path and unique segment ID for each participating node, i.e., $|\mathcal{C}| = 1$ and $|\bar{\Delta}_{I}| = 1, ~\forall \, I \in \{I_s, I_{j_1}, \ldots, I_{j_{h-1}}\}$. However, Bloom filters are also prone to false positive events, particularly if their parameters $m_1,~ m_2,~k_1$ and $k_2$ are not chosen appropriately. Within the scope of this work, the recovery process of spatial-provenance yields a false-positive event if the RSU identifies multiple solutions for the path and segment IDs that meet the conditions outlined in Section \ref{sec:Routing Constraints}, specifically, if $|\mathcal{C}| > 1$ or $|\mathcal{C}| = 1$ with $|\bar{\Delta}_{I}| > 1,$ for some $I \in \{I_s, I_{j_1}, \ldots, I_{j_{h-1}}\}$. Consequently, to get a single path and a unique segment ID for each participating node with a high probability, it is essential to carefully select the Bloom filter parameters $m_1,~ m_2,~k_1$, and $k_2$. Table \ref{tab:notation} presents a comprehensive list of notations in the protocol definition.

\subsection{Complexity Analysis}\label{sec:complexityAnalysis}
Recall that both BSE and TSE embed DDE in $\textbf{BF}_1$, incurring identical embedding cost $\mathcal{O}(\floor{h/2}k_1)$ along an $h$-hop path. The two schemes differ only in their use of $\textbf{BF}_2$, the embedding cost along an $h$-hop path is $\mathcal{O}(\floor{h/2}k_2)$ for BSE and $\mathcal{O}(\floor{(h+1)/3}k_2)$ for TSE. In contrast, the baseline LNE \cite[Section III-B]{ManishTDSC} incurs a embedding cost of $\mathcal{O}(hk_1)$ for $\textbf{BF}_1$ and $\mathcal{O}(hk_2))$ for $\textbf{BF}_2$.\looseness=-1 

During recovery, both schemes first reconstruct the traversed path from $\textbf{BF}_1$, incurring a computation cost of $\mathcal{O}(N^3k_1)$ at the RSU \cite[Table 1]{amogh}. Subsequently, the RSU verifies segment information using $\textbf{BF}_2$, which incurs a cost of $\mathcal{O}(h r \beta k_2)$ for BSE (refer to Section \ref{sec:BSE_recovery}) and $\mathcal{O}(h r \beta^2 k_2)$ for TSE (refer to Section \ref{sec:TSE_recovery}). In contrast, the baseline LNE \cite[Section III-B]{ManishTDSC} incurs a cost of $\mathcal{O}(N^2k_1)$ for $\textbf{BF}_1$ and $\mathcal{O}((h+1)rk_2)$ for $\textbf{BF}_2$. Compared to LNE \cite{ManishTDSC}, BSE and TSE trade a modest increase in computation at the RSU for a substantial reduction in node-level computation. Given that the RSU is computationally more powerful, this additional overhead has a marginal impact on overall performance. At the RSU, transient storage per received packet is also $m_1 + m_2$ bits. If the RSU stores all matching tuples for verification, the temporary storage requirement is $\mathcal{O}(hr\beta)$ for BSE and $\mathcal{O}(hr \beta^2)$ for TSE. Before proceeding to find the optimal parameters for our network model, we evaluate the relaxed-privacy attributes of our proposed methodology in the next section.\looseness=-1 

\begin{table} 
   \centering
   \caption{\label{tab:notation} Notations used in this paper.}
   \resizebox{6.5 cm}{!}{
   \begin{small}
   \begin{tabular}{|c|c|}
     \hline
     \textbf{Term} & \textbf{Meaning} \\
     \hline
     $N$ & Number of nodes in the network \\
     \hline
      $m_1$ & Size of $\textbf{BF}_1$ (in bits) \\
     \hline
     $k_1$ & Number of the Hash function used in $\textbf{BF}_1$ \\
     \hline
      $m_2$ & Size of $\textbf{BF}_2$ (in bits) \\
     \hline
     $k_2$ & Number of the Hash function used in $\textbf{BF}_2$ \\
     \hline
     $h$ & Number of hops traveled by the packet \\
     \hline
     $I_i$ & ID of $i^{th}$ node \\
     \hline
     $\beta$ & Communication constraint \\
     \hline
     $(x,y,z)$ & Double-edge connecting node $I_z \rightarrow I_y \rightarrow I_x$ \\
     \hline
     $\Delta, ~|\Delta|$ & Set of segment IDs, cardinality of set $\Delta$\\
     \hline
     $\alpha$ & Number of random bits lit in $\textbf{BF}_2$\\
     \hline
     $pid$ & Packet identity (ID)\\
     \hline
     $g ()$ & A function to translate $I_i \rightarrow \Delta$\\
     \hline
     $S()$ & A function to translate $\Delta \rightarrow [r]$\\
     \hline
     \end{tabular}
     \end{small}
     }
     \vspace{-0.3cm}
 \end{table}
\vspace{-0.3cm}
\subsection{Quantification of Relaxed-Privacy in BSE and TSE}

\bl{In this section, we quantify the relaxed-privacy attributes of our method using entropy-based measures. In our model, vehicles reveal the segment ID deterministically, where GPS coordinates are mapped to segment IDs, thereby concealing fine-grained locations. Unlike approaches such as differential privacy and geo-indistinguishability \cite{CCN}, which add noise to location data, our framework implicitly introduces uncertainty since a vehicle’s exact GPS position remains uniformly distributed within a segment. We quantify this privacy using entropy-based methods and relate it to geo-indistinguishability.\looseness=-1}


\bl{Let the RSU coverage area of $y$ square meters be divided into  $M \triangleq \frac{y}{\delta}$ Voronoi regions (as depicted in Fig. \ref{fig: network Image_car}), where $\delta$ is the GPS resolution in square meters. When a vehicle shares its exact GPS location, the RSU resolves an uncertainty of $\log_2(M)$ bits. In our proposed framework, the coverage area is partitioned into $r$ segments, each containing $M/r$ Voronoi regions. Therefore, the RSU recovers only $\log_2(r)$ bits, leaving a residual uncertainty of $\log_2(M/r)$ bits. This residual uncertainty is the measure of relaxed-privacy adopted in this work.}\looseness=-1

\bl{Privacy is quantified as the residual uncertainty of an observer about the precise GPS location of a vehicle after packet observation, and the information gain is denoted as $\Delta H$ bits.  Any forwarding node knows only that a packet originated from one of the $\beta$ upstream segments, giving a general prior uncertainty of $H^{prior} = \log_2(\beta M/r)$ bits. In LNE \cite{ManishTDSC}, no node explicitly shares its location with any neighbor. Therefore, $\Delta H =0$ bits at every forwarding node. In BSE (refer to Section \ref{sec:BSE_Embedding}), only $\lfloor h/2 \rfloor$ embedding nodes out of $h$ forwarding nodes explicitly receive a segment ID from an immediate neighbor, reducing their posterior uncertainty to $H^{post} =\log_2(M/r)$ bits, resulting $\Delta H = \log_2(\beta)$ bits at embedding node. The remaining $\lceil h/2 \rceil$ non-embedding nodes gain no information, i.e., $\Delta H = 0$ bits. A similar analysis can be carried out for TSE, and the resulting uncertainties are summarized in Table~\ref{tab:privacy_summary}. 
Across all three schemes, the RSU is able to recover only the segment IDs of forwarding nodes, resulting in $\Delta H = \log_2(r)$ bits. Non-authenticated nodes, which may act as eavesdroppers, have a wireless range assumed to be approximately $ 2\beta$ segments. Since segment IDs are embedded in Bloom filters and these nodes do not have $K_y$, their $\Delta H = 0$ bits across all three schemes. As observed from Row~3 of Table~\ref{tab:privacy_summary}, the privacy loss with respect to the neighboring nodes in TSE is twice that of BSE, while BSE incurs $\log_2(\beta)$ bits privacy loss with respect to the neighboring nodes compared to LNE \cite{ManishTDSC}.}\looseness=-1


\begin{table}[h]
\centering
{\color{black}
\caption{\bl{Privacy exposure of BSE and TSE relative to LNE \cite{ManishTDSC}.\looseness=-1}}
\label{tab:privacy_summary}
\footnotesize
\setlength{\tabcolsep}{4pt}
\renewcommand{\arraystretch}{1.3}
\resizebox{6.5 cm}{!}{
\begin{tabular}{|c|l|c|c|c|}
\hline
\textbf{No.} & \textbf{Observer / Quantity} & \textbf{LNE \cite{ManishTDSC}} & \textbf{BSE} & \textbf{TSE} \\
\hline
\multicolumn{5}{|l|}{\textit{Embedding node}} \\
\hline
1 & $H^{prior}$ (bits) 
  & $\log_2\!\left(\beta\frac{M}{r}\right)$ 
  & $\log_2\!\left(\beta\frac{M}{r}\right)$ 
  & $2\log_2\!\left(\beta\frac{M}{r}\right)$ \\
2 & $H^{post}$ (bits)  
  & $\log_2\!\left(\beta\frac{M}{r}\right)$ 
  & $\log_2\!\left(\frac{M}{r}\right)$ 
  & $2\log_2\!\left(\frac{M}{r}\right)$ \\
3 & $\Delta H$ (bits)  
  & $0$ 
  & $\log_2(\beta)$ 
  & $2\log_2(\beta)$ \\
\hline
\multicolumn{5}{|l|}{\textit{Non-embedding forwarding node}} \\
\hline
4 & $H^{prior}$ (bits) 
  & $-$ 
  & $\log_2\!\left(\beta\frac{M}{r}\right)$ 
  & $\log_2\!\left(\beta\frac{M}{r}\right)$ \\
5 & $H^{post}$ (bits)  
  & $-$ 
  & $\log_2\!\left(\beta\frac{M}{r}\right)$ 
  & $\log_2\!\left(\beta\frac{M}{r}\right)$ \\
6 & $\Delta H$ (bits)  
  & $-$ 
  & $0$ 
  & $0$ \\
\hline
\multicolumn{5}{|l|}{\textit{RSU}}  \\
\hline
7 & $H^{prior}$ (bits) 
  & $\log_2(M)$ 
  & $\log_2(M)$ 
  & $\log_2(M)$ \\
8 & $H^{post}$ (bits)  
  & $\log_2\!\left(\frac{M}{r}\right)$ 
  & $\log_2\!\left(\frac{M}{r}\right)$ 
  & $\log_2\!\left(\frac{M}{r}\right)$ \\
9 & $\Delta H$ (bits)  
  & $\log_2(r)$ 
  & $\log_2(r)$ 
  & $\log_2(r)$ \\
\hline
\end{tabular}}}
\end{table}

Our privacy framework can also be interpreted through the lens of geo-indistinguishability \cite{CCN}. Each GPS Voronoi area, denoted by $x$, is deterministically mapped to its corresponding segment $A_i$, for $i \in [r]$, and upon observing $A_i$, the RSU’s posterior becomes a uniform distribution over all GPS Voronoi areas within that segment. Consequently, for any two locations $x,x'$ within the same segment $A_i$, their posterior distributions are identical which results in zero Kullback–Leibler (KL) divergence; hence there is zero loss of geo-indistinguishability within the segment. However, for $x \in A_i$ and $x' \in A_j$ for $i \neq j$, the KL divergence is infinite, leading to maximum loss of geo-indistinguishability. In this case, the RSU accurately learns the segment IDs of the vehicles, without knowing the precise GPS Voronoi areas of the vehicles inside each segment. The binary behavior of the KL divergence, i.e., taking only values of zero or infinity, can be easily modeled using entropy-based uncertainty measures, which are sufficient to fully capture the privacy properties of the proposed framework, without loss of generality.\looseness=-1

In the following section, we comprehensively study the false-positive events linked to our methodology and discuss optimizing the Bloom filter parameters for BSE and TSE.\looseness=-1
\vspace{-0.1cm}
\section{Optimization of Bloom Filter Parameters}  \label{sec:optimization} 
\vspace{-0.1cm}
As discussed in Section \ref{sec:recovery}, we now formally analyze the selection of parameters $m_1$, $m_2$, $k_1$, and $k_2$ to minimize the rate of false positive events. When employing the proposed spatial-provenance recovery approach, let $E_B^{(1,2)}$ for BSE and $E_T^{(1,2)}$ for TSE denote a false positive event in which the RSU identifies multiple solutions for the path and segment IDs that meet the constraints outlined in Section \ref{sec:Routing Constraints}.
With such events, the RSU will be unable to ascertain the actual path taken by the packet or the true segment IDs of the forwarding nodes. To provide high reliability in the spatial-provenance recovery process, the probability of false positive events, denoted by $\Pr{(E_B^{(1,2)})}$ for BSE and $\Pr{(E_T^{(1,2)})}$ for TSE, must be minimized over the number of packets traversed through the network. Let $\mathcal{G}_{\beta}$ denote the collection of all valid mappings $g(\cdot)$ between the IDs of the forwarding nodes and $\Delta$ satisfying the constraints in Section \ref{sec:Routing Constraints}. Consequently, before implementing the proposed protocols for a specific hop-length, it is essential to select the parameters $m_1,~ m_2,~k_1$ and $k_2$ for $\textbf{BF}_{1}$ and $\textbf{BF}_{2}$ to minimize the above probability of false positive events, averaged across all valid spatial distributions of the forwarding nodes, denoted as $g \in \mathcal{G}_{\beta}$, that comply with the constraints outlined in Section \ref{sec:Routing Constraints}.\looseness=-1

The recovery procedure, outlined in Section \ref{sec:recovery}, entails path restoration from $\textbf{BF}_1$ and subsequently the recovery of segment IDs from $\textbf{BF}_2$. Therefore, quantifying the probability of false positive events is inherently complex because it involves two Bloom filters. Given that the Bloom filter for embedding double-edge is comprehensively studied in \cite{amogh},\cite{suraj}, we focus on solving a variant of the problem wherein we assume that $\textbf{BF}_1$ is optimized for embedding double-edge and only one path is retrieved from $\textbf{BF}_1$. When recovering the segment IDs for BSE, let $E_B^{(1)}$ represent the false positive event from $\textbf{BF}_1$ in which multiple paths of hop-length $h$ are retrieved. Consequently, the RSU declares an event $E_B^{(1, 2)}$ on the occurrence of $E_B^{(1)}$; otherwise, it continues to verify the segment IDs of all the forwarding nodes retrieved from $\textbf{BF}_{1}$. While using a path retrieved from $\textbf{BF}_{1}$ in BSE recovery, let $E_B^{(2)}$ represent the event in which more than one segment ID sequence adhering to the constraints outlined in Section \ref{sec:Routing Constraints} are retrieved from $\textbf{BF}_{2}$. Similarly, $E_T^{(1)}$, $E_T^{(1, 2)}$ and $E_T^{(2)}$ denotes the corresponding false positive events in TSE recovery process, analogous to the definition of  $E_B^{(1)}$, $E_B^{(1, 2)}$ and $E_B^{(2)}$ in BSE recovery process. Therefore, with $E_{B}^{(1)}$, $E_{B}^{(2)}$ for BSE and using $E_T^{(1, 2)}$ and $E_T^{(2)}$ for TSE, the overall probability of false positive event $\Pr(E_{B}^{(1,2)})$ for BSE and $\Pr(E_{T}^{(1,2)})$ for TSE are bounded by
\begin{IEEEeqnarray*}{rcl}
\Pr(E_{B}^{(1,2)}) \leq \Pr(E_{B}^{(1)}) + (1-\Pr(E_{B}^{(1)})) \times \Pr(E_{B}^{(2)}),\\
\Pr(E_{T}^{(1,2)}) \leq \Pr(E_{T}^{(1)}) + (1-\Pr(E_{T}^{(1)})) \times \Pr(E_{T}^{(2)}),
\end{IEEEeqnarray*}
where $\Pr(E_{B}^{(1)})$ represents the probability of false positive event in $\textbf{BF}_1$, whereas $\Pr(E_{B}^{(2)})$ signifies the probability of false positive event in $\textbf{BF}_2$ conditioned on no false positive event from $\textbf{BF}_{1}$ in BSE recovery process. Note that if the size of $\textbf{BF}_1$ is sufficiently large, and the number of Hash functions within it is already optimized\footnote{\bl{Bloom filters for embedding the double-edges of a path have been thoroughly investigated for large networks in \cite{amogh, suraj}. Thus, solutions for optimizing the number of Hash functions $k_{1}$ may be identified in \cite{amogh, suraj}. By using proven double-edge embedding techniques, our approach applies to large networks, as $\textbf{BF}_{1}$ absorbs the scalability feature.}}, then $\Pr(E_{B}^{(1)}) \approx 0$. Consequently, we may approximate $\Pr(E_{B}^{(1,2)}) \approx \Pr(E_{B}^{(2)})$ under such assumptions. Similarly, $\Pr(E_{T}^{(1)})$ and $\Pr(E_{T}^{(2)})$ follows for TSE recovery process analogous to the $\Pr(E_{B}^{(1)})$ and $\Pr(E_{B}^{(2)})$ of BSE recovery process, and with the similar arguments $\Pr(E_{T}^{(1,2)}) \approx \Pr(E_{T}^{(2)})$ under similar assumptions as explained above.\looseness=-1

Utilizing the relaxation as mentioned above, we present Problem \ref{problem2} that optimizes the parameters of $\textbf{BF}_{2}$.\looseness=-1
{\begin{problem_stmt}{}\label{problem2}
Given $N$, $h$, $\Delta$, $m_1$, $m_2$ and $\beta$ solve:
\begin{IEEEeqnarray}{rCl}
k^{+}_2 = \arg \min_{\{k_2\}} \underset{g\in \mathcal{G}_{\beta}}{\mathbb{E}} [\Pr(E_{B}^{(2)})],
\end{IEEEeqnarray}
\begin{IEEEeqnarray}{rCl}
k^{-}_2 = \arg \min_{\{k_2\}} \underset{g\in \mathcal{G}_{\beta}}{\mathbb{E}} [\Pr(E_{T}^{(2)})].
\end{IEEEeqnarray}
\end{problem_stmt}}

For clarity and conciseness in our explanation, we introduce a unified notation to discuss both methods, i.e., BSE and TSE, which share similar data structures, i.e., $\textbf{BF}_1$ and $\textbf{BF}_2$ along with almost similar recovery procedures on a broad level. We denote the terms related to false positive events as $E_X$ and related probability notation with $X$ subscript instead of $B$ or $T$, where $X \in \{B, T\}$, with $B$ corresponding to BSE and $T$ corresponding to TSE. This notation allows us to present a common theoretical framework while distinguishing between the two protocols, thereby streamlining the explanation without sacrificing precision.\looseness=-1
\vspace{-0.2cm}
\section{Optimization of $\textbf{BF}_2$ Parameters} \label{sec:optimization_of_BF2}
In this section, we will initially obtain the expression for $\underset{g\in \mathcal{G}_{\beta}}{\mathbb{E}} [\Pr(E_{X}^{(2)})]$ by characterizing how false positives arise from both spatial path ambiguity and Bloom-filter collisions. In particular, Section \ref{sec:closed_form_objective_fucntion} derives the main expression for $\underset{g\in \mathcal{G}_{\beta}}{\mathbb{E}} [\Pr(E_{X}^{(2)})]$ by averaging the probability of false positives over all valid segment sequences, which are enumerated using a $\beta$-ary tree, and by quantifying collision-induced verifications at the RSU. Thereafter, Section \ref{sec:Low_complexity_C_j_any_beta} analyzes how Bloom-filter false positives translate into segment sequence ambiguity by counting, for a given true segment sequence, the number of alternative segment sequences that can be recovered due to one or more unintended recoveries. Finally, Section \ref{sec:solving_problem_2} combines these results to obtain a closed-form analytical expression for the optimal solution of Problem \ref{problem2}.
\vspace{-0.3cm}
\subsection{Closed-form Expression for $\underset{g\in \mathcal{G}_{\beta}}{\mathbb{E}} [\Pr(E_{X}^{(2)})]$ } \label{sec:closed_form_objective_fucntion}
To solve Problem \ref{problem2}, we need to determine $\underset{g\in \mathcal{G}_{\beta}}{\mathbb{E}} [\Pr(E_{X}^{(2)})]$. 
Assume that for a specific spatial distribution $g \in \mathcal{G}_{\beta}$ of forwarding nodes, a single path with hop-length $h$ has been successfully reconstructed from $\textbf{BF}_1$, and the nodes along this path have embedded their segment IDs in $\textbf{BF}_2$ for the RSU to ascertain the locations of the forwarding nodes. Based on the first principle,\looseness=-1
\begin{IEEEeqnarray}{rCl}
\label{eqn:obejctive1}
    \underset{g\in \mathcal{G}_{\beta}}{\mathbb{E}} [\Pr(E_{X}^{(2)})] = \sum\limits_{g \in \mathcal{G}_{\beta}} \Pr(g)\Pr(E_{X}^{(2)}|g),
\end{IEEEeqnarray}
\noindent \noindent where $\Pr(E_{X}^{(2)}|g)$ represents the probability of false positive event from $\textbf{BF}_2$ conditioned that spatial distribution of the forwarding nodes is $g$, and $\Pr(g)$ represents the probability of $g$. Assuming $\Pr(g) = \frac{1}{|\mathcal{G}_{\beta}|}$, we write
\begin{IEEEeqnarray}{rCl}
\label{eqn:objective2}
    \underset{g\in \mathcal{G}_{\beta}}{\mathbb{E}} [\Pr(E_{X}^{(2)})] = \frac{1}{|\mathcal{G}_{\beta}|}\sum\limits_{g \in \mathcal{G}_{\beta}} \Pr(E_{X}^{(2)}|g).
\end{IEEEeqnarray}

Given that every node which is embedding the segment ID in $\textbf{BF}_2$ utilizes $k_{2}$ Hash functions, the total number of bits lit in $\textbf{BF}_{2}$ during the packet's traversal is a random variable with a maximum value of $min(m_2, k_2h^*_X)$, where $h^*_X$ for $X \in \{B,T\}$, represents the number of nodes which embeds the segment IDs in $\textbf{BF}_2$.\looseness=-1 

Formally, in BSE protocol (refer to Section \ref{sec:BSE_Embedding}), every alternate node performs the insertion operation of segment IDs in $\textbf{BF}_2$; therefore, the value of $h^*_{B}$ is given by\looseness=-1
\begin{IEEEeqnarray}{rCl}
\label{eqn:h^*_B}
    h^*_{B} =\floor{h/2},
\end{IEEEeqnarray}
where $\floor{.}$ denotes the floor function. Whereas in TSE protocol (refer to Section \ref{sec:TSE_Embedding}), every third node performs the insertion operation of segment IDs in $\textbf{BF}_2$; therefore, the value of  $h^*_{T}$ is 
\begin{IEEEeqnarray}{rCl}
\label{eqn:h^*_T}
    h^*_{T} =\floor{(h+1)/3}.
\end{IEEEeqnarray}
With $\alpha$ representing the number of bits lit in $\textbf{BF}_2$ by all the forwarding nodes, each term in the summation given in the RHS of \eqref{eqn:objective2} is given in the following proposition.\looseness=-1
\begin{proposition} \label{prop:Pr_Efp} The expression for $\Pr(E_{X}^{(2)}|g)$, can be written using Bayes' theorem as:
\begin{IEEEeqnarray}{rCl}
\label{Eq:Pr_Efp2}
    \Pr(E_{X}^{(2)}|g) = \sum\limits_{\alpha=1}^{\min(m_{2},k_{2}h^*_X)}\Pr(E_{X}^{(2)}|\alpha,g)\Pr(\alpha),
\end{IEEEeqnarray}
where $\Pr(E_{fp}^{(2)}|\alpha,g)$ represents the probability of false positive event conditioned on $\alpha$ bits being lit in $\textbf{BF}_{2}$, and $\Pr(\alpha)$ signifies the probability of $\alpha$ bits lit in $\textbf{BF}_2$, as given in \cite{suraj}:\looseness=-1
\begin{IEEEeqnarray}{rCl}
\label{eqn:Pr_alpha}
\Pr(\alpha)=  \frac{{m_2 \choose \alpha}  \sum\limits_{\gamma=0}^{\alpha}(-1)^{\gamma} {{\alpha}\choose{\gamma}}(\alpha-\gamma)^{k_2 h^*_X}}{m_2^{k_2 h^*_X}}.
\end{IEEEeqnarray}
\end{proposition}
\begin{proof}
    The proof is analogous to the one presented in \cite{suraj}.
\end{proof}
By substituting \eqref{Eq:Pr_Efp2} into \eqref{eqn:objective2}, the objective function of Problem \ref{problem2} can be expressed as:\looseness=-1
\begin{small}
\begin{IEEEeqnarray}{rCl}
\label{eqn:objective5}
\sum\limits_{\alpha=1}^{\min(m_{2},k_{2}h^*_X)}  \Pr(\alpha) \bigg[\frac{1}{|\mathcal{G}_{\beta}|} \sum\limits_{g \in \mathcal{G}_{\beta}}  \Pr(E_{X}^{(2)}|\alpha,g)\bigg].
\end{IEEEeqnarray}
\end{small}
Here, \eqref{eqn:objective5} represents the probability of false positive events, averaged over all valid segment sequences $\mathcal{G}_{\beta}$ that satisfy the system constraints.
Subsequently, in the rest of the section, we present an expression for $\frac{1}{|\mathcal{G}_{\beta}|} \sum\limits_{g \in \mathcal{G}_{\beta}} \Pr(E_{X}^{(2)}|\alpha,g)$.\looseness=-1  

Computing the average requires enumerating all valid segment sequences, as defined below.\looseness=-1

\begin{definition}\label{def:valid_sequecnes}
A valid sequence of segments is a sequence from $\Delta$ wherein the indices of the segments are arranged in non-decreasing order, subject to the constraint that the difference between consecutive indices is not more than $\beta$.\looseness=-1 
\end{definition}
For example, with a hop-length $h = 5$, $\beta = 3$, and $|\Delta| = 7$, a valid segment sequence is $\{A_{1}, A_{2}, A_{3}, A_{5}, A_{6}, A_{7}\}$, whereas $\{A_{1}, A_{2}, A_{3}, A_{3}, A_{7}, A_{7}\}$ is not valid.
We define the collection of all valid sequences of subscripts of segment IDs as $\mathcal{P}_{\beta}$. Observe that there exists a one-to-one mapping between the set $\mathcal{G}_{\beta}$ and the set $\mathcal{P}_{\beta}$, whereby one represents the mapping of node IDs to segment IDs, while the other denotes the ordered sequence of subscripts of segment IDs. Consequently, under the communication constraint $\beta$, the equality $|\mathcal{G}_{\beta}| = |\mathcal{P}_{\beta}|$ holds. The next lemma delineates the result regarding the cardinality of $\mathcal{P}_{\beta}$ as a function of $|\Delta|,~\beta$, and $h$.\looseness=-1  

\begin{lemma} \label{lemma:path}
We can enumerate the elements of $\mathcal{P_{\beta}}$ for given $h$, $|\Delta|=r$ and $\beta$, by constructing a $\beta$-ary tree.\looseness=-1
\end{lemma}
\begin{proof}
The set $\mathcal{P}_\beta$ of valid $h$-hop segment sequences is obtained by enumerating all root-to-leaf paths of a $\beta$-ary tree that has depth $h$ and has the constraint that leaf node values do not exceed $r$. This completes the construction.
\end{proof}
 
\begin{corollary} \label{lemma:pathbeta}
Using $\mathcal{P}_{\beta}$ enumerated from Lemma \ref{lemma:path}, we can easily determine $|\mathcal{P_{\beta}}|$ by counting arguments.\looseness=-1 
\end{corollary}

Before calculating the expression for $\Pr(E_{X}^{(2)}|\alpha,g)$, we introduce the necessary definitions below.\looseness=-1 
\begin{definition}\label{def:bi-segment tuple}
A bi-segment tuple refers to a three-element tuple $(w,y,z)$, where $w \in [N] \backslash \{N-1\} \cup \{s\}$ is the node ID; $y,z \in \Delta $ are the segment IDs, where $S(y) \leq S(z),~ |S(z)-S(y)|\leq \beta$. In BSE, a bi-segment tuple represents the relationship between a node ID and the corresponding segment IDs embedded in $\textbf{BF}_2$.\looseness=-1 
\end{definition}
\begin{definition}\label{def:tri-segment tuple}
A tri-segment tuple refers to a four-element tuple $(w,x,y,z)$, where $w \in [N] \backslash \{N-1\} \cup \{s\}$ is the node ID; $x,y,z \in \Delta $ are the segment IDs, where $S(x) \leq S(y) \leq S(z),~ |S(z)-S(y)|\leq \beta$ and $|S(y)-S(x)|\leq \beta$. A tri-segment tuple represents the relationship between a node ID and the corresponding segment IDs embedded in $\textbf{BF}_2$.\looseness=-1 
\end{definition}
The probability of false positive events further depend on the Bloom filter size and the occurrence of collision events. 

\begin{definition}\label{def:collision} A collision event in a Bloom filter is defined as an event where we recover an object from the Bloom filter that is not initially embedded onto the Bloom filter.\looseness=-1
\end{definition}
The probability of a collision event in a Bloom filter of size $m$ bits, with $\alpha$ bits lit, is expressed as $(\frac{\alpha}{m})^{k}$, where $k$ is the number of Hash functions used to embed each item, assuming that Hash functions are independent and uniformly select bit positions in $m$.  For $\textbf{BF}_{2}$, given $m_2$, and $k_2$, the probability of a collision event is $p_1 = \left(\frac{\alpha}{m_2}\right)^{k_2}$, and the probability of non-collision events is $p_2 = 1 - p_{1}$. The collision events are essential in enumerating the false positive events in $\textbf{BF}_{2}$. For example, a packet traverses a network where $h = 4$, $\beta = 2$, and $|\Delta| = 7$ and the path followed is $\{(I_6,A_{6}) \rightarrow (I_4,A_{4}) \rightarrow (I_3,A_{3}) \rightarrow (I_2,A_{2}) \rightarrow (I_1,A_{1}) \}$. Using BSE, only two bi-segment tuples: $(I_4, A_{4}, A_{6})$ and $(I_2,A_{2}, A_{3})$ are embedded in $\textbf{BF}_2$. During the recovery of a path at the RSU, a false positive event may occur if the bi-segment tuple $(I_4, A_5, A_6)$, which was not embedded, is recovered from $\textbf{BF}_2$. In this scenario, the RSU has an alternative path $\{A_{1}, A_{2}, A_{3}, A_{5}, A_{6}\}$ along with the embedded path, which satisfies the constraints in Section \ref{sec:Routing Constraints}. Under similar conditions, using TSE, only one tri-segment tuple $(I_4,A_{3}, A_{4}, A_{6})$ is embedded in $\textbf{BF}_2$. During the recovery of a path at the RSU, a false positive event may occur if there is a recovery of the tri-segment tuple $(I_4,A_3, A_5, A_6)$ from $\textbf{BF}_2$, which was not embedded during TSE.\looseness=-1

Based on Section \ref{sec:Bloom Filters}, during the recovery process from $\textbf{BF}_2$, the RSU verifies all the possible bi-segment tuples in BSE recovery, while in TSE recovery, it verifies all the possible tri-segment tuples. Next, we will enumerate the number of bi-segment and tri-segment tuples the RSU verifies.\looseness=-1 

\begin{proposition} \label{prop:sigma_B}
In BSE, the number of bi-segment tuples, denoted as $\sigma_{B}$, is given by:\looseness=-1 
\begin{IEEEeqnarray}{rcl}\label{eq:sigma_B}
   \sigma_{B} = \floor{\frac{h}{2}} \sum \limits_{p_{x_i}=1}^{r}\bigl(\min(r,p_{x_i}+\beta) -p_{x_i} +1\bigl).
\end{IEEEeqnarray}
\end{proposition}
\begin{proof}
    The proof is provided in Appendix \ref{apx:sigma_B}.
\end{proof}
\vspace{-0.3cm}
\begin{proposition} \label{prop:sigma_T}
In TSE, the number of tri-segment tuples, denoted as $\sigma_{T}$, is given by:\looseness=-1 

\begin{small}
\begin{IEEEeqnarray}{rcl}\label{eq:sigma_T}
   \sigma_{T} =
   \left\{ 
\begin{array}{ll}

            \floor{\frac{h+1}{3}}\left( \sum \limits_{p_{x_i}=1}^{r} H(p_{x_i})\right), & \hspace{-2em} (h+1)~mod~ 3 \neq 0,\\

            \left(\floor{\frac{h+1}{3}}- 1 \right)\left( \sum \limits_{p_{x_i}=1}^{r} H(p_{x_i})\right) &   \hspace{-1em}\\
            + \sum\limits_{p_{x_2}=1}^{\min(1 + \beta, r)} 
            (\min(p_{x_2} + \beta, r) - p_{x_2} + 1), &  \hspace{-0.5em}(h+1)~ mod~3 = 0,
\end{array}
\right.
\end{IEEEeqnarray}
\end{small}

\noindent \noindent where \begin{small}$H(p_{x_i})=(p_{x_{i}}-\max(1, p_{x_{i}}-\beta)+1 )( \min(p_{x_{i}}+\beta, r)-p_{x_{i}}+1)$.\end{small}
\end{proposition}
\begin{proof}
    The proof is provided in Appendix \ref{apx:sigma_T}.
\end{proof}

Hereafter, we collectively term bi-segment and tri-segment tuples as segment tuples and use specific terminology only when distinction is needed. Based on the above propositions, during the recovery process, the RSU verifies $\sigma_{X}$ segment tuples, for $X \in \{B,T\}$, out of which only $h^*_X$ segment tuples were embedded in $\textbf{BF}_2$ during the embedding process. We define the remaining segment tuples, $\sigma_{X}-h^*_X$ in number, as false segment tuples with a formal definition provided below.\looseness=-1

\begin{definition} \label{def:setF}
For a given valid segment sequence represented in $\mathcal{P}_{\beta}$, the set of segment tuples that were not initially embedded in $\textbf{BF}_2$, however, is checked for membership in $\textbf{BF}_2$ at the RSU, is called the set of false segment tuples, denoted as $\mathcal{F}_X$ where $X \in \{B,T\}$ represents BSE and TSE respectively.\looseness=-1
\end{definition}

For BSE and TSE, the cardinality of $\mathcal{F}_{B}$ and $\mathcal{F}_{T}$, are respectively given by
\begin{IEEEeqnarray}{rcl}\label{eq:F_B}
|\mathcal{F}_{B}| = \sigma_{B}- \floor{\frac{h}{2}}, ~~|\mathcal{F}_T| = \sigma_T - \floor{\frac{h+1}{3}}.
\end{IEEEeqnarray}
We want to underscore that $\mathcal{F}_X$ will vary for each sequence in $\mathcal{P}_{\beta}$. Nonetheless, $|\mathcal{F}_X|$ remains constant independent of the sequence in $\mathcal{P}_{\beta}$.\looseness=-1
\begin{definition} \label{def:setR}
For a given valid segment sequence represented in $\mathcal{P}_{\beta}$, the set of segment tuples retrieved from $\textbf{BF}_2$, which are not part of the embedded segment sequence, is referred to as extra recoveries, represented by $\mathcal{R}_X$ where $X \in \{B,T\}$ represents BSE and TSE respectively. Note that $\mathcal{R}_X \subseteq \mathcal{F}_X$ trivially.\looseness=-1
\end{definition}
While $|\mathcal{F}_X|$ remains constant for every sequence in $\mathcal{P}_{\beta}$, the set $\mathcal{R}_X, ~X \in \{B,T\}$ constitutes a random subset of $\mathcal{F}_X$, contingent upon the collision event defined in Definition \ref{def:collision}. Next, an example is presented for $\mathcal{F}_X$ and $\mathcal{R}_X$.\looseness=-1
\begin{example}
Consider a path traversed by a packet $[(I_{1},A_{1})\rightarrow (I_{2},A_{2})\rightarrow (I_{3},A_4)\rightarrow (I_{4},A_5) \rightarrow (I_{5},A_5),\rightarrow (I_{6},A_5),\rightarrow (I_{7},A_7)]$, for network parameters $\beta =2$, $h=6$ and $r=7$. During packet traversal, the bi-segment tuples which are embedded in $\textbf{BF}_2$ are $\{(I_{2},A_{2}, A_{4} ),(I_{4},A_{5},A_{5}),(I_{6},A_5,A_7)\}$. For this path, $\mathcal{F}_B = \{I_2,I_4,I_6\}\times \{(A_i,A_j)~|~A_i,A_j \in \Delta,1\leq i \leq j \leq \min(i+\beta,r), |j-i|\leq \beta\} \backslash \{(I_{2},A_{2}, A_{4} ),(I_{4},A_{5},A_{5}),(I_{6},A_5,A_7)\}$. If $\mathcal{R}_B=\{(I_4,A_7, A_7)\}$ or $\mathcal{R}_B=\{(I_4,A_2,A_3)\}$, this will not result in a false positive event in BSE recovery, since substituting the extra recovery in the original path would violate the node communication constraint in Section \ref{sec:Routing Constraints}. However, if $\mathcal{R}_B = \{(I_4,A_4,A_5)\}$ or $\mathcal{R}_B = \{(I_4,A_4,A_4)\}$, then replacing this with $(I_4, A_5,A_5)$ will lead to a false positive event. For similar path traversal using TSE, we have $\mathcal{F}_T =\{I_3,I_6\} \times \{(A_x,A_y,A_z) ~|~ A_x,A_y,A_z \in \Delta; 1\leq x\leq y\leq z\leq \min(y+\beta,r); |y-x|,|z-y|,|z-x|\leq \beta\} \backslash \{(I_3,A_2,A_4,A_5), (I_6,A_5,A_5,A_7)\}$, where the recovery of $\mathcal{R}_T= \{(I_6,A_5,A_6,A_7)\}$ will cause a false positive event.\looseness=-1
\end{example}

\begin{table}[h!] 
\caption{ Illustrating the computation of false positive events as a function of extra recovery, where $X \in \{B,T\}$.}
\centering
\resizebox{8.8 cm}{!}{
 \begin{tabular}{|c | c |  c | c| c|} 
 \hline
 Valid Paths  & Extra 1 Recovery  & Extra 2 Recovery  & \ldots & Extra $|\mathcal{F}|$ Recovery\\
 ($\mathcal{P}_{\beta}$) & ($\mathcal{R}_X =1$) & ($\mathcal{R}_X =2$) & & ($\mathcal{R}_X =|\mathcal{F}_X|$)
 \\[0.5ex] 
 \hline\hline
 1111\ldots$1_{h}$ & $C_{1,1_X}$ & $C_{1,2_X}$ & \ldots & $C_{1,|\mathcal{F}_X|_{X}}$ \\ 
 1111\ldots$2_{h}$ & $C_{2,1_X}$ & $C_{2,2_X}$ & \ldots &  $C_{2,|\mathcal{F}_X|_{X}}$ \\
 . & . & . & \ldots & . \\
 . & . & . & \ldots & . \\
. & . & . & \ldots  & . \\
 123\ldots $r$ \ldots$r_{h}$ & $C_{|\mathcal{P}_{\beta}|,1_X}$ & $C_{|\mathcal{P}_{\beta}|,2_X}$ & \ldots &  $C_{|\mathcal{P}_{\beta}|,|\mathcal{F}_X|_{X}}$\\ [1ex] 
 \hline\hline
 Total & $C_{1_X} = \sum_{x=1}^{|\mathcal{P}_{\beta}|}C_{X,x,1} $ &  $C_{2_X} = \sum_{x=1}^{|\mathcal{P}_{\beta}|}C_{x,2_X} $ & \ldots &  $C_{(|\mathcal{F}_X|_{X})}= \sum_{x=1}^{|\mathcal{P}_{\beta}|}C_{x,|\mathcal{F}_X|_{X}}$\\
 \hline
 \end{tabular}}
 \label{T:falsepositiveCalculation}
\end{table}
For a given segment sequence within $\mathcal{P}_{\beta}$, a false positive event happens when any other sequence in $\mathcal{P}_{\beta}$ is retrieved from $\textbf{BF}_2$. As other sequences in $\mathcal{P}_{\beta}$ are obtained through extra recoveries, it is necessary to correlate these extra recoveries with the false positive events. The number of extra recoveries, which is $|\mathcal{R}_X|$, during the recovery process from $\textbf{BF}_{2}$ is a random variable. Consequently, we must evaluate all potential scenarios of extra recoveries such that $|\mathcal{R}_X| \leq |\mathcal{F}_X|$. Let $\mathcal{P}_{\beta} = \{\mathbf{p}_{1}, \mathbf{p}_{2}, \ldots, \mathbf{p}_{|\mathcal{P}_{\beta}|}\}$. Assume that $\mathbf{p}_{x}$, for some $x \in [1, |\mathcal{P}_{\beta}|]$, represents the valid sequence of subscripts of segment IDs and sequence $\mathbf{p}_{x}$ is embedded in $\textbf{BF}_2$ by $h^*_X$ nodes. Let $C_{X,x,j}$ denote the number of sequences in $\mathcal{P}_{\beta}\backslash \{\mathbf{p}_{x}\}$ that can be retrieved from $\textbf{BF}_2$, given that $j = |\mathcal{R}_X|$ extra recoveries are lit in $\textbf{BF}_2$. If $C_{X,x,j}$ can be calculated for all $j \in [1, |\mathcal{F}_X|]$, then the probability of false positive events for the sequence $\mathbf{p}_{x}$ may be determined. Let $C_{X,j} \triangleq \sum_{x = 1}^{|\mathcal{P}_{\beta}|}C_{X,x,j}$ denotes the sum of all $C_{X,x,j}$ values over every sequence in $\mathcal{P}_{\beta}$ for a specified number of extra recovery $|\mathcal{R}_X|=j$. Subsequently, we employ the following theorem to calculate the probability of a false positive event using $C_{X,j}$, for $j \in [1,|\mathcal{F}_X|]$.\looseness=-1
\begin{theorem} \label{Theorem-Main_eq}
The average probability of the false positive event from \eqref{eqn:objective5}, given $g$ and $h$, can be expressed as:\looseness=-1

\begin{small}\begin{IEEEeqnarray*}{rcl}
\label{overall_fp2_theorem}
\frac{1}{|\mathcal{G}_{\beta}|}\sum\limits_{g \in \mathcal{G}_{\beta}}\Pr(E_{X}^{(2)}|\alpha,g) = \frac{1}{|\mathcal{G}_{\beta}|}\sum_{{j}=1}^{|\mathcal{F}_X|} p_1^{j}\times p_2^{|\mathcal{F}_X|-{j}}\times \left(\sum_{x = 1}^{|\mathcal{G}_{\beta}|}C_{X,x,j}\right),
\end{IEEEeqnarray*}\end{small}

\noindent\noindent where $p_1 = (\alpha/m_2)^{k_2}$ and $p_2=1-p_1$, $|\mathcal{G}_{\beta}|$ is taken from Lemma \ref{lemma:path} and $|\mathcal{F}_X|$ is taken from \eqref{eq:F_B} for BSE or TSE.\looseness=-1
\end{theorem}
\begin{proof} The proof is provided in Appendix \ref{apx:Theorem-Main_eq}.\end{proof}
Utilizing the above result, the expression for $\Pr(E_{X}^{(2)}|\alpha,g)$ can be computed provided $C_{X,j} \triangleq \sum_{x = 1}^{|\mathcal{P}_{\beta}|}C_{X,x,j}$ values are taken from Table \ref{T:falsepositiveCalculation}. Subsequently, we provide a technique to populate the values in Table \ref{T:falsepositiveCalculation}. To fill in the values in Table \ref{T:falsepositiveCalculation}, it is not feasible to get the closed-form expression for $C_{X,x,j}$. 
Next, we will introduce low-complexity algorithms to populate Table \ref{T:falsepositiveCalculation}, where we calculate the number of false but valid segment sequences due to  $\textbf{BF}_2$.\looseness=-1 

\vspace{-0.2cm}
\subsection{A Method to Compute $C_{X,x,j}$ for BSE and TSE}
\label{sec:Low_complexity_C_j_any_beta}
In this section, we propose a methodology to calculate the value of $C_{X,x,j}$, for $X \in\{B,T\}$, $x \in [1,|\mathcal{P}_{\beta}|]$ and $j\geq 1$. 
First, we propose low-complexity algorithms to calculate the exact values for $j=1$, i.e., $C_{X,x,1}$, for $x \in [1, |\mathcal{P}_{\beta}|]$, which captures the impact of an extra recovery $j=1$ from $\textbf{BF}_2$ on the segment sequence ambiguity at the RSU.
After that, using the exact values of $C_{X,x,1}$, we propose a closed-form expression to obtain the lower bound on the values of $C_{X,x,j}$, for $j\geq 2$ and $x \in [1, |\mathcal{P}_{\beta}|]$.\looseness=-1


\begin{proposition}\label{prop:value_of_C_Xx1}
Given that $j=1$ for $\mathbf{p}_x,~x\in[1,|\mathcal{P}_{\beta}|]$, one can compute $C_{X,x,1}$ using an algorithm through counting arguments.\looseness=-1
\end{proposition}
\begin{proof} The proof is provided in Appendix \ref{apx:prop:value CXx1}.
\end{proof}

Further, we will explain the method of computing a lower bound on the values of $C_{X,x,j}$ for $j \geq 2$.\looseness=-1
\begin{proposition}\label{prop:value extension any beta}
Given there are $|\mathcal{R}_X|=j, j\geq 2,~X\in\{B,T\}$, extra recoveries for a given valid segment sequence $\mathbf{p}_x,~x\in[1,|\mathcal{P}_{\beta}|]$, the value of  $C_{X,x,j}$, for $j\geq 2$, is given by:\looseness=-1 
\begin{IEEEeqnarray}{rCl}
\label{eq:cxj_beta_expression}
C_{X,x,j} \geq \sum_{l=1}^{C_{X,x,1}}\binom{|\mathcal{F}_X| - l}{j-1}.
\end{IEEEeqnarray}
where $|\mathcal{F}_X|$ is obtained from Definition \ref{def:setF},    $C_{X,x,1}$ is obtained from Proposition \ref{prop:value_of_C_Xx1}.\looseness=-1 
\end{proposition}
\begin{proof} The proof is provided in Appendix \ref{apx:prop:value extension any beta}.
\end{proof}
Using the above result, we are now ready to have a lower bound on the expression for $\{C_{X,j}, j > 1\}$, given by\looseness=-1
\begin{IEEEeqnarray}{rCl}
\label{eq:cj_beta_expression1}
C_{X,j} \geq \sum_{x=1}^{|\mathcal{P}_{\beta}|} \sum_{l=1}^{C_{X,x,1}}\binom{|\mathcal{F}_X| - l}{j-1}.
\end{IEEEeqnarray}
From the above lower bound, it is clear that as long as we have $C_{X,x,j}$, we can compute $\{C_{X,j}\}$ for any $\beta$. The complexity of calculating \eqref{eq:cj_beta_expression1} is $\mathcal{O}(|\mathcal{P}_{\beta}|C_{X,x,1})$. For BSE, by plugging the values of $C_{B,x,1}$ from Proposition \ref{prop:value_of_C_Xx1} in \eqref{eq:cxj_beta_expression}, we obtain $C_{B,x,j}$, which is used in \eqref{overall_fp2} to obtain a lower bound on the expression for $\underset{g\in \mathcal{G}_{\beta}}{\mathbb{E}} [\Pr(E_{B}^{(2)}|a,g)]$. Then the obtained lower bound on $\underset{g\in \mathcal{G}_{\beta}}{\mathbb{E}} [\Pr(E_{B}^{(2)}|a,g)]$ is plugged in \eqref{eqn:objective5} to obtain a lower bound on the overall average probability of false-positive event for BSE. Similarly, for TSE, by plugging the value of $C_{T,x,1}$ from Proposition \ref{prop:value_of_C_Xx1} in \eqref{eq:cxj_beta_expression}, we obtain a lower bound on $C_{T,x,j}$, which is used in \eqref{overall_fp2} to obtain a lower bound on the expression for $\underset{g\in \mathcal{G}_{\beta}}{\mathbb{E}} [\Pr(E_{T}^{(2)}|a,g)]$. Then the obtained lower bound on $\underset{g\in \mathcal{G}_{\beta}}{\mathbb{E}} [\Pr(E_{T}^{(2)}|a,g)]$ is plugged in \eqref{eqn:objective5} to obtain a lower bound on the overall average probability of false-positive event for TSE.\looseness=-1

In this section, we presented a method for calculating a lower bound on $C_{X,j}$ for any $\beta$. Next, we use this lower bound to derive the solution for Problem \ref{problem2}.\looseness=-1 
\vspace{-0.2cm}
\subsection {On Solving Problem \ref{problem2}}
\label{sec:solving_problem_2}
For BSE, by plugging the values of $C_{B,x,1}$ from Proposition \ref{prop:value_of_C_Xx1} into \eqref{eq:cxj_beta_expression}, we obtain a lower bound on $C_{B,x,j}$. Now, using the value of $C_{B,x,1}$ from Proposition \ref{prop:value_of_C_Xx1} and the obtained value of $C_{B,x,j}$ along with the value of $|\mathcal{F}_B|$ from \eqref{eq:F_B} into \eqref{overall_fp2}, we obtain a lower bound on the expression for $\underset{g\in \mathcal{G}_{\beta}}{\mathbb{E}} [\Pr(E_{B}^{(2)}|\alpha,g)]$. Finally, we use this lower bound to obtain a lower bound on the average probability of false positive events. To obtain a near-optimal solution to Problem \ref{problem2} for BSE, we employ the computed lower bound on $\underset{g\in \mathcal{G}_{\beta}}{\mathbb{E}} [\Pr(E_{B}^{(2)})]$ as the objective function to determine the values around $k_2^{+}$ using a gradient descent algorithm. Likewise for TSE, upon determining a lower bound on $\underset{g\in \mathcal{G}_{\beta}}{\mathbb{E}} [\Pr(E_{T}^{(2)})]$; we employ a gradient descent algorithm to determine the near-optimal value of \eqref{eqn:objective5} by varying the parameter $k_2$. In the subsequent part, we will offer results to validate the efficacy of the solution mentioned above.\looseness=-1

\section{Experimental Results } \label{sec:results}
In this section, we will present results on the near-optimality of analytical expressions of BSE and TSE discussed in Section \ref{sec:optimization_of_BF2}.  Subsequently, we will also present the performance of our methods with respect to the baseline LNE \cite{ManishTDSC, manish_WCNC}, followed by some results on practical aspects of BSE and TSE.\looseness=-1

\subsection{False Positive Rate Analysis of BSE and TSE}
\bl{To validate the accuracy of the proposed analytical expressions, we compare the expression in Section \ref{sec:Low_complexity_C_j_any_beta} with the probability of false positive events produced by simulation results for various values of $h$, $\beta$, $|\Delta|$, $m_1$, and $m_2$.\footnote{\bl{Since existing wireless standards do not reserve dedicated bits for provenance, information related to spatial provenance must be accommodated within the packet payload. Therefore, for our experiments, $m_1$ and $m_2$ are chosen arbitrarily to drive down the probability of false-positive events to a very small value. In particular, the choice of $m_1$ follows prior work in \cite{amogh,suraj}, which studied it with respect to network density and payload budget.\looseness=-1}}} Henceforth, throughout the section, false positive rates refer to the probability of false positive events of spatial provenance recovery. To generate the results, we utilize a network with parameters $N=11, ~h = 10$, $|\Delta| = 3$, and $\beta=1$. The dimensions of $\textbf{BF}_1$ are selected to provide minimal false positive rates of the order of $10^{-4}$, thereby fulfilling the criterion $\Pr(E_{X}^{(1,2)}) \approx \Pr(E_{X}^{(2)})$, as discussed in Section \ref{sec:optimization}. Henceforth, we represent the probability of a false positive event $\Pr(E_{X}^{(2)})$ as $P_{X}, ~X \in \{B,T\}$, where $B,~T$ represent BSE and TSE, respectively. For generating simulation results, $\textbf{BF}_2$ of size $m_2=30$ bits is utilized. \bl{For each value of $k_2$, $10^9$ packets are transmitted across the multihop network to obtain $P_X$ values. Note that the number of packets in the simulation is kept large to validate the analytical expressions rather than to emulate realistic network traffic.} The associated $P_{X}$ values are depicted as a function of $k_2$ in Fig. \ref{fig:LSE_LDE_h_10_delta_3_beta_1}. Fig. \ref{fig:LSE_LDE_h_10_delta_3_beta_1} illustrates that the simulation curve and the analytical lower bound attain their minima around similar values and at closely located $k_2$ values. Comparable tests are performed for TSE, demonstrating that the analytical expressions for TSE also provide good approximations of the outcomes produced by simulation results. Analogous experiments are performed by varying the network parameters with the outcomes shown in Fig. \ref{fig:LSE_LDE_h_8_delta_3_beta_3}, Fig. \ref{fig:LSE_LDE_h_8_delta_7_beta_3}, and Fig. \ref{fig:LSE_LDE_h_8_delta_11_beta_3}. The analysis reveals that TSE outperforms BSE for identical network parameters with respect to $P_{X}$. The graphs presented in Fig. \ref{fig:LSE_LDE_initial} also demonstrate that the minimum $P_{T}$ occurs at a higher $k_2$ compared to the minimum $P_{B}$, which occurs at a lower $k_2$. This observation suggests that $P_{X}$ benefit is accompanied by a slight increase in computational complexity, which is directly associated with the optimal number of Hash functions $k_2$ required at each embedding node.\looseness=-1
\begin{figure}[ht!]
\vspace{-0.3cm}
\centering
\begin{subfigure}{0.49\columnwidth} 
\caption{}
\label{fig:LSE_LDE_h_10_delta_3_beta_1}
    \includegraphics[trim={0.3cm 0 1.4cm 1cm},clip,scale=0.3, width =\textwidth]{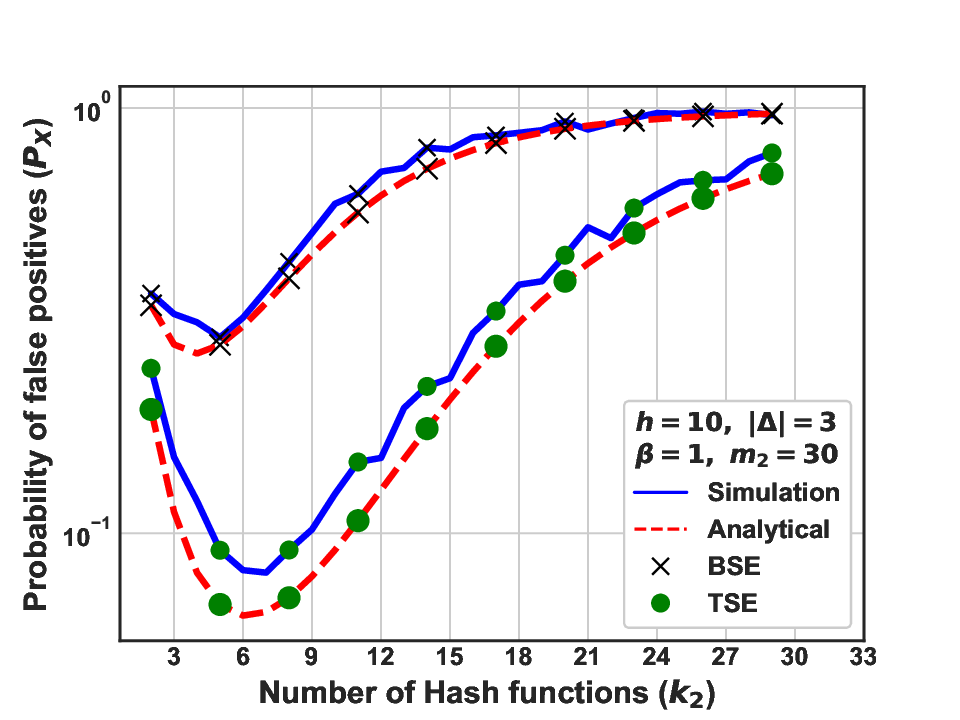}   
\end{subfigure}
\begin{subfigure}{0.49\columnwidth}
\caption{}
\label{fig:LSE_LDE_h_8_delta_3_beta_3}
    \includegraphics[trim={0.3cm 0 1.4cm 1cm},clip,scale=0.3, width =\textwidth]{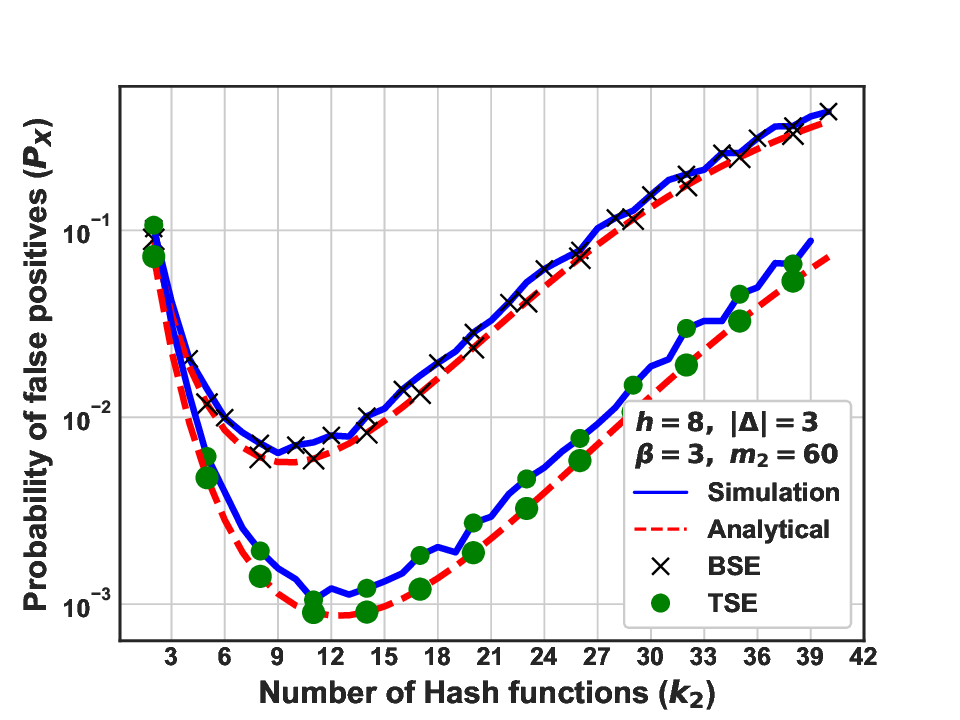}
    
\end{subfigure} 
\begin{subfigure}{0.49\columnwidth}
\caption{}
\label{fig:LSE_LDE_h_8_delta_7_beta_3}
    \includegraphics[trim={0.3cm 0 1.4cm 1cm},clip,scale=0.3, width =\textwidth]{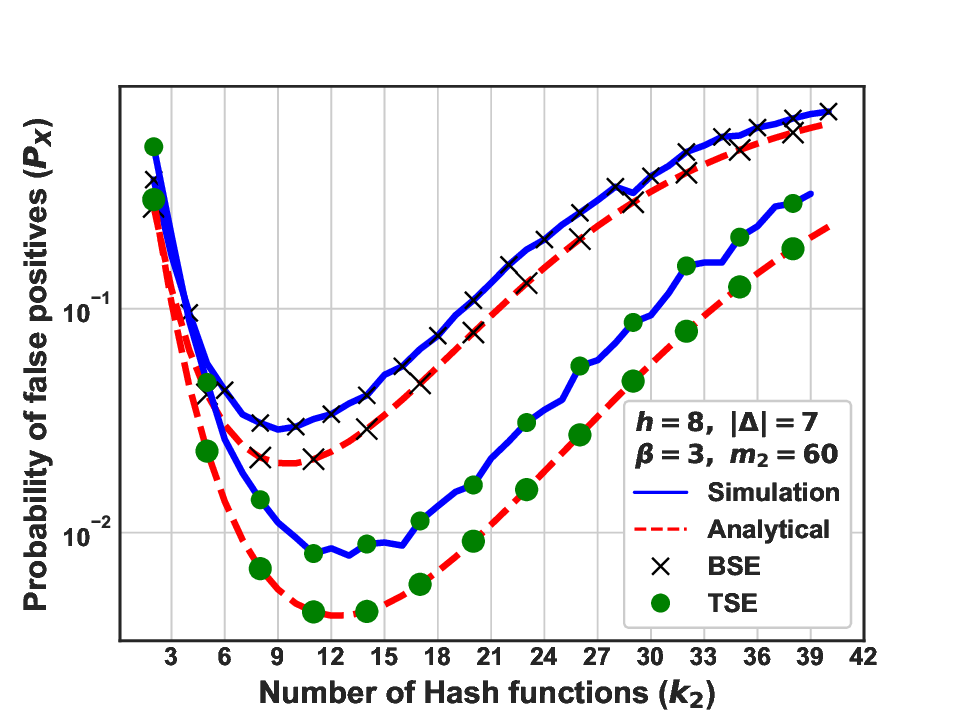}
     
\end{subfigure}  
\begin{subfigure}{0.49\columnwidth}

\caption{}
\label{fig:LSE_LDE_h_8_delta_11_beta_3}
    \includegraphics[trim={0.3cm 0 1.4cm 1cm},clip,scale=0.3, width =\textwidth]{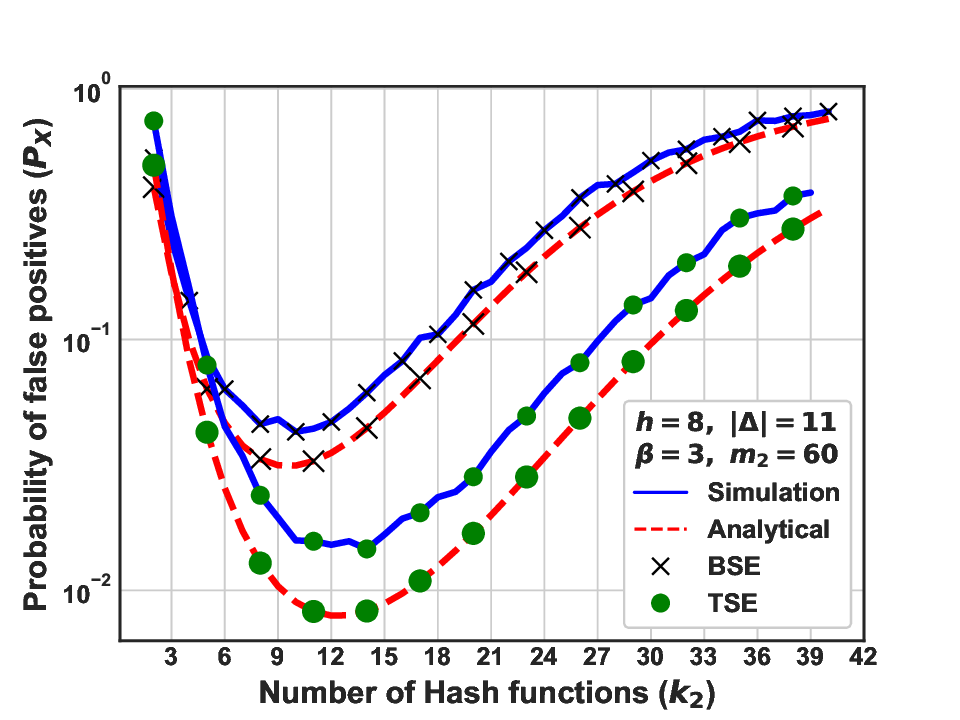}
    
\end{subfigure} 
\caption{Comparison of the analytical bound with the simulation results for BSE and TSE for the network setting of (a) $N=11,~\beta =1,~m_2=30, ~|\Delta| =3 $, (b) $N=9,~\beta =3,~m_2=60, ~|\Delta| =3 $, (c) $N=9,~\beta =3,~m_2=60, ~|\Delta| =7 $, (d) $N=9,~\beta =3,~m_2=60, ~|\Delta| =11$.\looseness=-1} 
\vspace{-0.3cm}
\label{fig:LSE_LDE_initial}
\end{figure}

\begin{figure}
\centering
\begin{subfigure}{0.49\columnwidth} 
\caption{}
 \label{fig:Tradeoff_delta_Vs_pfp}
    \includegraphics[trim={1cm 0.2cm 1.5cm 0.9cm},clip,scale=0.32, width =\textwidth]{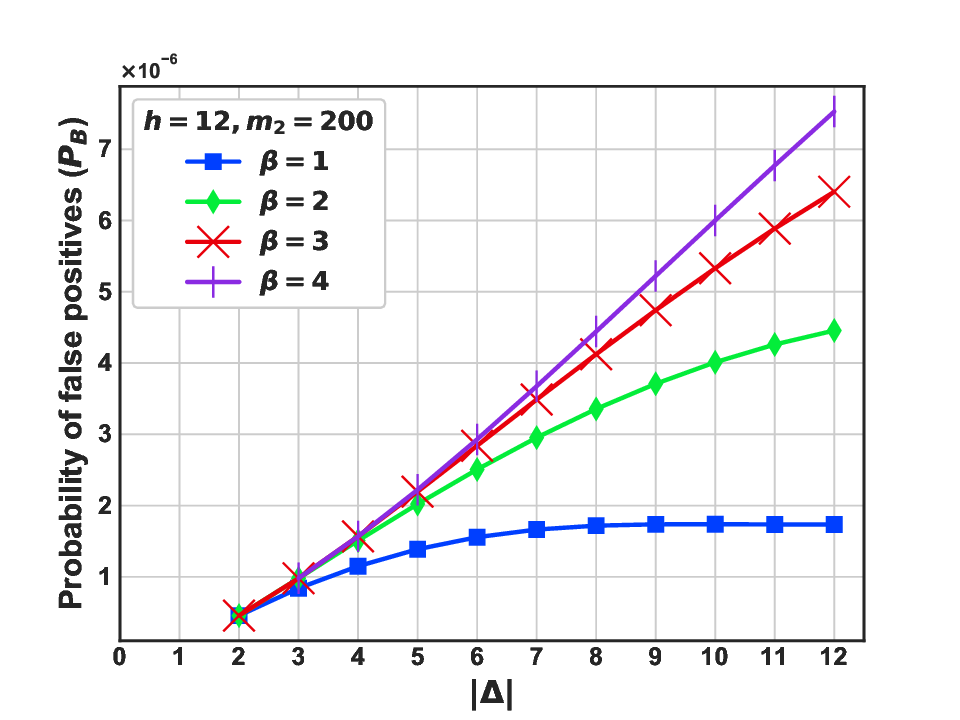}   
\end{subfigure}
\begin{subfigure}{0.49\columnwidth}
\caption{}
\label{fig:Tradeoff_Delta_Vs_m2}
     \includegraphics[trim={0.5cm 0.2cm 1.5cm 1cm},clip,scale=0.32, width =\textwidth]{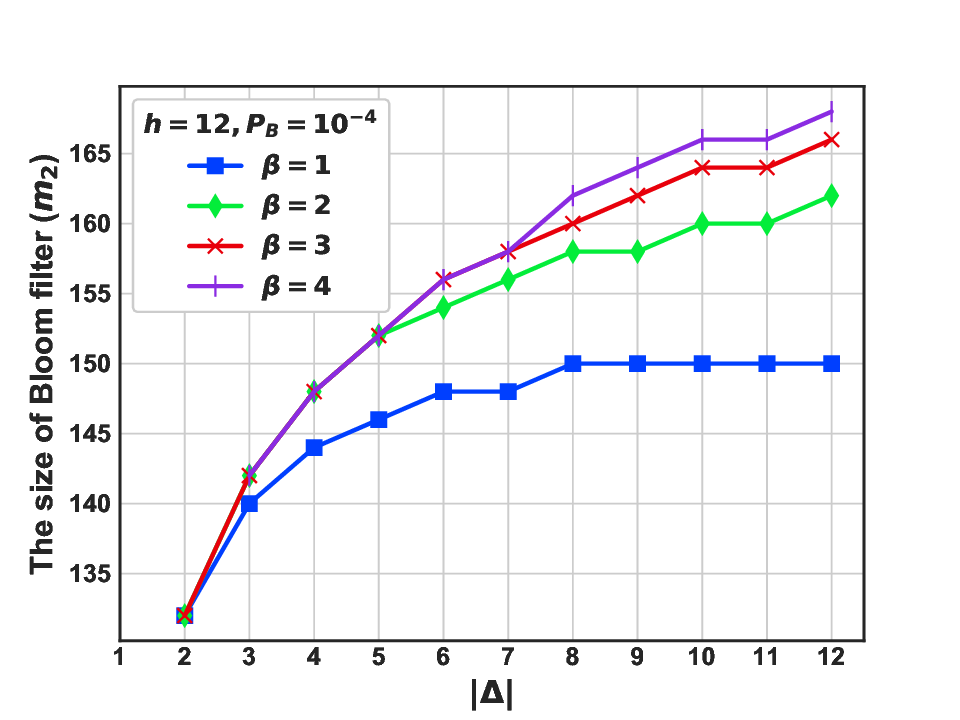}  
\end{subfigure} 
\begin{subfigure}{0.49\columnwidth} 
\caption{}
 \label{fig:Tradeoff_delta_Vs_pfpLDE}
    \includegraphics[trim={1cm 0.2cm 1.5cm 0.9cm},clip,scale=0.32, width =\textwidth]{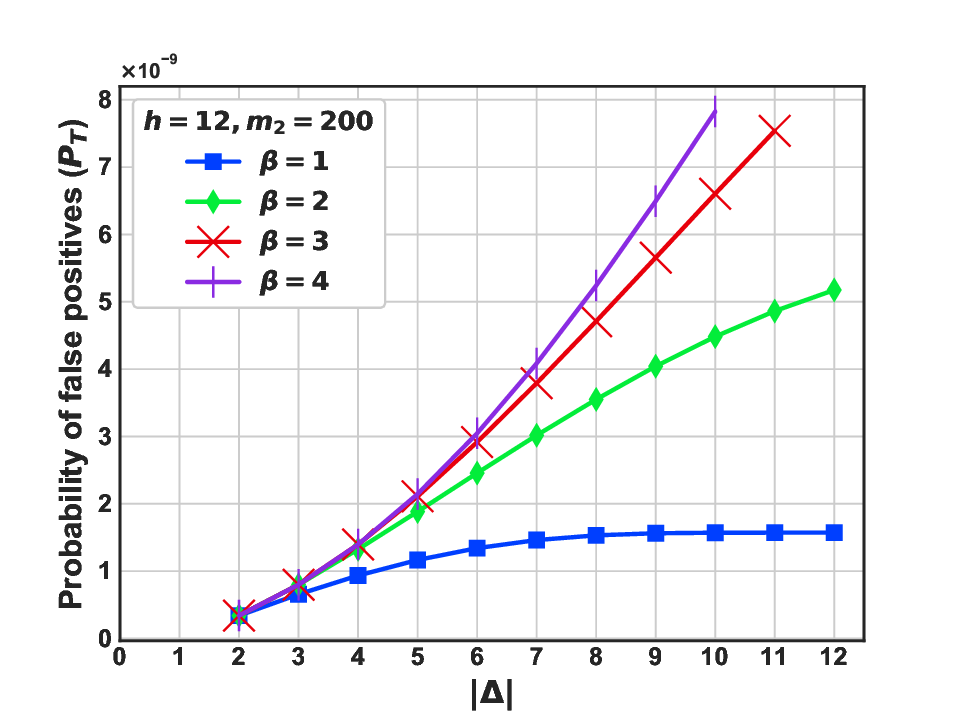}   
\end{subfigure}
\begin{subfigure}{0.49\columnwidth}
\caption{}
\label{fig:Tradeoff_Delta_Vs_m2LDE}
     \includegraphics[trim={0.5cm 0.2cm 1.5cm 1cm},clip,scale=0.32, width =\textwidth]{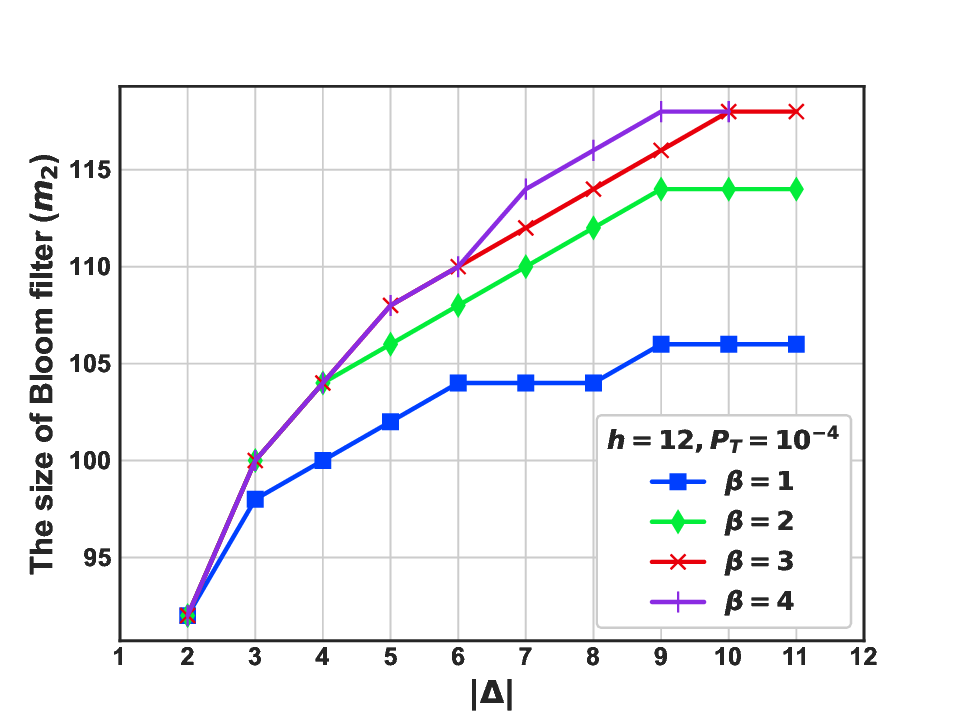}  
\end{subfigure}  
\caption{Analysis of the correlation between $P_{X}$ and $|\Delta|$ for a network configuration with $h = 12$ and $m_2 = 200$ bits (a) BSE, (c) TSE. Analysis of the relationship between the required size of $\textbf{BF}_2$ and $|\Delta|$ for a network setup with $h = 12$ and a fixed $P_{X} = 10^{-4}$. The magnitude represented on the y-axis is measured in bits (b) BSE, (d) TSE.\looseness=-1} 
\label{fig:Tradeoff}
\vspace{-0.4cm}
\end{figure}
We examine the influence of privacy on $P_{X}$ as a function of $\beta$. In this study, with optimized $\textbf{BF}_1$, we fix $m_2=200$ bits for a fixed hop network with $h=12$ and $N=13$. For a specified $\beta$ and constant $|\Delta|$, namely $\beta = 1$ and $|\Delta| = 2$, we transmit a large number of packets across a multihop network for $h=12$. We vary the number of Hash functions within the range $2\leq k_2\leq m_2$ and record the minimum $P_{X}$ for this network configuration, plotting the results against $|\Delta|$ in Fig. \ref{fig:Tradeoff_delta_Vs_pfp}. We conduct a similar experiment for $2 \leq |\Delta| \leq 12$ and plot the corresponding minimum $P_{X}$ values against each $|\Delta|$ for a specified $\beta$. We conduct identical experiments utilizing varying values of $\beta$, the results of which are also plotted in Fig. \ref{fig:Tradeoff_delta_Vs_pfp}. When the coverage area of the RSU is divided into $|\Delta| = 2$ segments, the spatial privacy of the nodes is higher compared to the case when the coverage area of the RSU is divided into $|\Delta| = 12$ segments, where the spatial privacy of the nodes is lower. For a specific $\beta$, representing the transmission power of the nodes, i.e., how far a node can transmit, Fig. \ref{fig:Tradeoff_delta_Vs_pfp} illustrates that $P_{X}$ increases with $|\Delta|$. An increase in $P_{X}$ corresponding to an increase in $|\Delta|$ results from the augmentation of segment tuples with the increasing $|\Delta|$. In summary, for a fixed size $\textbf{BF}_2$ and a specified number of nodes $N$, relaxing location privacy results in poor reliability of recovery due to an increase in minima values of $P_{X}$. Similar experiments are performed for TSE, results of which are depicted in Fig. \ref{fig:Tradeoff_delta_Vs_pfpLDE}, which also shows similar behavior.\looseness=-1 

We also examine the trade-off between the required size of $\textbf{BF}_2$ and $|\Delta|$. With optimized $\textbf{BF}_1$, we fixed $P_{B} = 10^{-4}$ for the multihop network of $h=12$ with $N=13$. For given $\beta$, we vary the number of segments from $|\Delta| = 2$ to $|\Delta| = 12$. For each $|\Delta|$, we determine the required size of $m_2$ by transmitting a large number of packets from the source node to the RSU, and plot the corresponding $m_2$ as a function of $|\Delta|$ in Fig. \ref{fig:Tradeoff_Delta_Vs_m2}. Fig. \ref{fig:Tradeoff_Delta_Vs_m2} illustrates that increased spatial resolution ($|\Delta|$) necessitates a large number of $\textbf{BF}_2$ bits ($m_2$) to maintain $P_{B}=10^{-4}$. Similar tests are conducted for TSE, and the results illustrate analogous behavior, as shown in Fig. \ref{fig:Tradeoff_Delta_Vs_m2LDE}.\looseness=-1 

\begin{figure}
    \centering
    \includegraphics[trim={0cm 0.0cm 1.4cm 0.5cm},clip,scale=0.4]{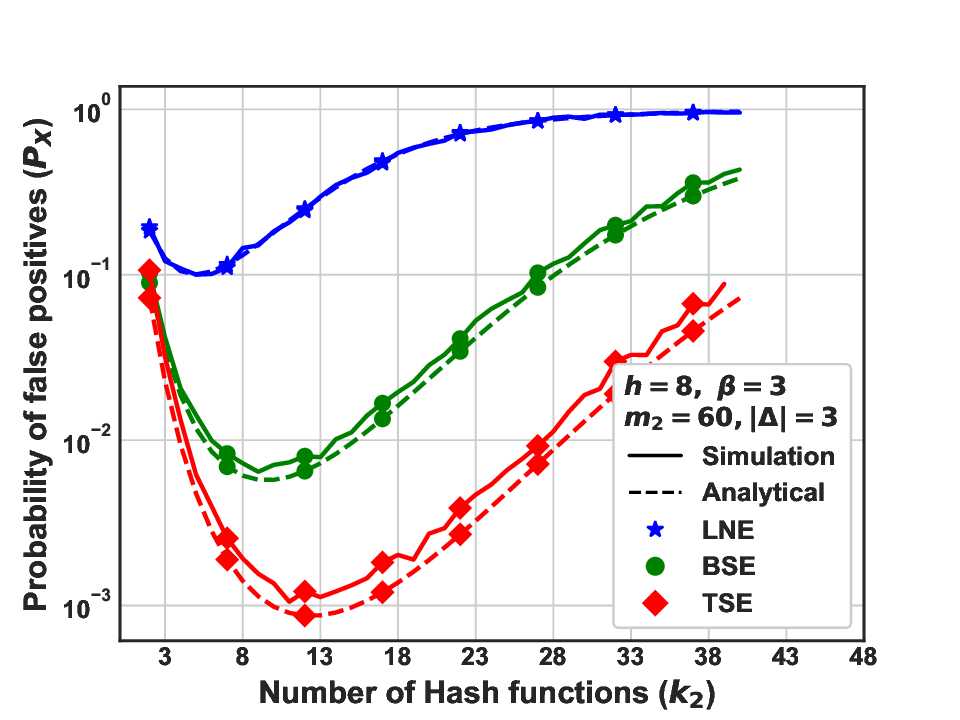}
    \caption{Comparison of $P_{X}$ for LNE \cite{ManishTDSC}, BSE and TSE for the network parameters $N = 9, ~\beta = 3, m_2=60$ with $|\Delta| = 3$ and $h=8$.}\looseness=-1
    \label{fig:LN_LSE_LDE_Comparison}
    \vspace{-0.5cm}
\end{figure}

For a baseline comparison, we evaluate the efficacy of our protocols against LNE in \cite{ManishTDSC, manish_WCNC}. For comparison, we simulate a network with $h=8$, $\beta=3$, and $|\Delta| =3$. Utilizing the optimized $\textbf{BF}_1$, we select $m_2=60$ bits. A large number of packets are transmitted from the source node to the RSU for each value of $k_2$. The associated $P_{X}$ values are plotted as a function of $k_2$ in Fig. \ref{fig:LN_LSE_LDE_Comparison}. The graph demonstrates that, for identical network parameters, TSE significantly outperforms BSE and LNE \cite{ManishTDSC} with respect to $P_{X}$. 
This is because of fewer number of embedding of segment tuples in TSE as compared to BSE and LNE, i.e., for hop-length $h$, there are $h$ embeddings in $\textbf{BF}_2$ for LNE \cite{ManishTDSC}, $\floor{\frac{h}{2}}$ embedding in case of BSE, whereas $\floor{\frac{h+1}{3}}$ embedding in TSE. Therefore, $\textbf{BF}_2$ in TSE is sparsely filled compared to BSE and LNE \cite{ManishTDSC}, leading to a lower collision rate, resulting in lower $P_{X}$. It is to be noted that the benefit in $P_{X}$ that we observe for TSE comes at the cost of privacy, where two neighboring nodes explicitly share their segment ID with an embedding node. Additionally, sharing the segment ID by neighbor nodes with an embedding node results in communication overhead (see Section \ref{sec:Testbed_timimgs}).\looseness=-1

\begin{figure}
    \centering
    \includegraphics[trim={0.1cm 0 1.4cm 0.9cm},clip,scale=0.4]{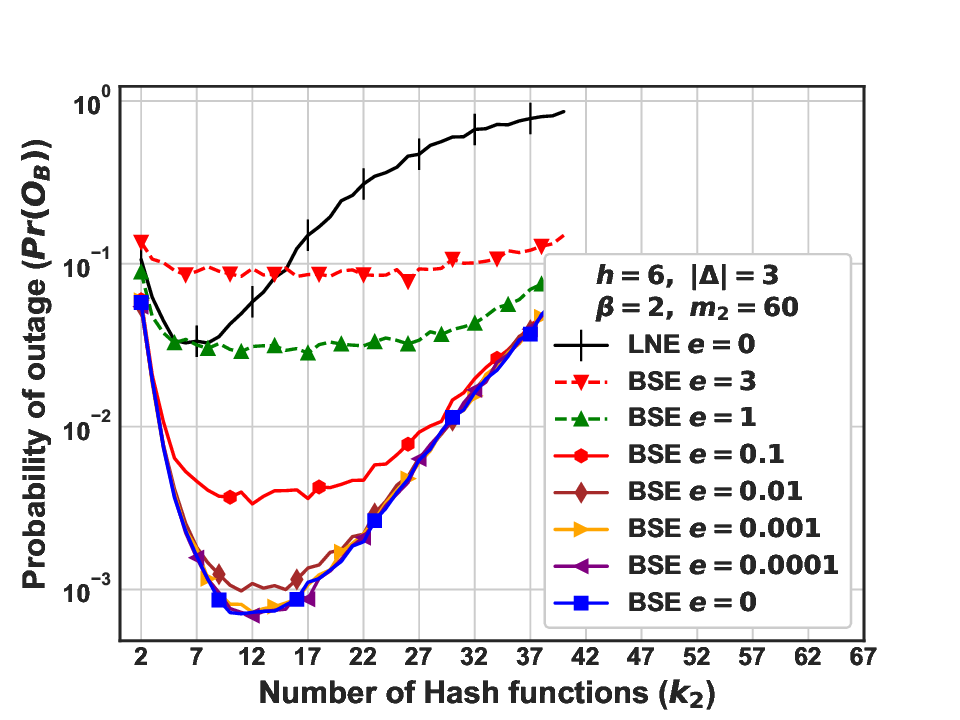}
    \caption{Comparison of $P_{X}$ for BSE with varying error values on a network of $N = 7, ~\beta = 2$ with $|\Delta| = 3$ and $h=6$.\looseness=-1}
    \label{fig:LSE_error_plot}
    \vspace{-0.5cm}
\end{figure}

\subsection{Analysis with Neighbors' Imperfect Segment Information} \label{sec:Imperfect_neighbor}
As defined in Section \ref{sec:network_model}, our network model assumes that the neighboring nodes explicitly share their segment IDs directly with the embedding nodes without error, thereby enabling their segment IDs to be communicated to the RSU. However, we assume a scenario of stringent privacy, where the neighbor nodes in the network do not explicitly share their segment ID with the embedding nodes. Therefore, to incorporate the suggested protocols, i.e., BSE or TSE, to support the low-latency constraint of the network, an embedding node needs to know the segment ID of its neighbors. Since the segment IDs are not shared directly, an embedding node must infer them indirectly. This segment ID inference requires knowledge of the relative positions of neighboring nodes. To determine these relative positions, we propose the embedding node to employ RSSI-based distance estimation techniques \cite{RSSI_distance, RSSI_Survey_indoor} to estimate the approximate distance of its neighbors. Subsequently, an embedding node will use the estimated distances in conjunction with its own GPS coordinates to infer the segment IDs of its neighbors with the help of the dictionary disseminated by the RSU. However, RSSI-based distance estimation techniques have an associated error \cite{RSSI_Error} in distance estimation, which may result in an error in inferring the segment ID of the node. In the rest of this section, we study the performance of BSE and TSE when the embedding nodes employ the above mentioned physical-layer methods for localizing their neighbors.\looseness=-1  

At an embedding node, the inference of the segment IDs of the neighboring nodes is facilitated by the RSSI-based techniques defined above. Let $e$ denote the probability of error in inferring the segment IDs of neighbors at an embedding node. When the packet is recovered at the RSU, the RSU is said to be in \emph{state of outage} if it is unable to accurately infer the original segment sequence taken by the packet. This occurs in either of the following scenarios: (i) when the RSU identifies a false positive, regardless of whether the embedded segment IDs contain errors; or (ii) when the RSU recovers a single segment sequence that is in error, i.e., it does not match the actual segment sequence traversed by the packet. The probability of outage for BSE, denoted by $\Pr(O_B)$, represents the fraction of packets for which the RSU is in the outage state. In a $h$-hop path for BSE, the segment ID of every alternate node is inferred; therefore, the maximum number of segment IDs that may be in error is $\floor{h/2}$. Assuming that an identical RSSI-based segment ID inference technique is employed at the embedding nodes and that the error probability of inferring segment IDs, $e$, remains constant across the embedding nodes, the probability of a recovered segment sequence being in error is $\Pr(E)= 1-(1-e)^{\floor{h/2}}$. False positive event ($E^{(2)}_B$), refer Section \ref{sec:optimization}, and error event ($E$) are independent, i.e., $\Pr(E \cap E^{(2)}_B) = \Pr(E) \Pr(E^{(2)}_B)$; therefore, the probability of outage $\Pr(O_B)$ is given by $\Pr(O_B) = 1 - (1-e)^{\floor{h/2}}(1-P_{B}),$ where $P_{B}$ is the probability of a false positive event for BSE without error. Similarly, for TSE, the probability of outage for TSE, $\Pr(O_T)$, is given by $\Pr(O_T) = 1 - (1-e)^{2\floor{(h+1)/3}}(1-P_{T}),$ where $P_{T}$ is the probability of a false positive event for TSE without error.\looseness=-1 

In the rest of this section, we compare two approaches: BSE/TSE under relaxed-privacy conditions, where neighboring nodes explicitly share their segment ID with an embedding node, and BSE/TSE under strict privacy constraints, where the embedding nodes infer the segment IDs of neighboring nodes through RSSI. We conduct an experiment for BSE, utilizing a network configuration with $h=6$, $|\Delta|=3$ and $\beta=2$. With optimized $\textbf{BF}_1$, we use $m_2=60$ bits for BSE. For a specified $e$, we transfer a large number of packets from the source node to the RSU and obtain the value of $\Pr(O_B)$. Fig. \ref{fig:LSE_error_plot} presents $\Pr(O_B)$ as a function of $k_2$. We repeat the same experiment with different values of error probability $e$ and plot them in Fig. \ref{fig:LSE_error_plot}. From Fig. \ref{fig:LSE_error_plot}, it is observed that BSE, assisted by RSSI-based techniques, may accommodate an error probability of up to $e=0.1\%$ in neighbor segment ID inference without substantially affecting $\Pr(O_B)$. However, when the error probability increases to $e=1\%$, the performance of BSE, assisted by RSSI-based techniques, closely resembles with LNE technique \cite{ManishTDSC}. This demonstrates a significant deterioration in accuracy and efficiency as evidenced by an increase in $\Pr(O_B)$. This experiment helps to identify a threshold of $e$ beyond which the performance of BSE assisted by RSSI-based techniques deteriorates significantly. Similar experiments, when conducted for TSE, revealed that beyond the threshold error probability of $0.01\%$, the performance of TSE closely resembles that of BSE.\looseness=-1

\subsection{Testbed Based Timing Analysis of BSE and TSE Techniques} \label{sec:Testbed_timimgs}
This section provides a thorough comparison of the communication overhead incurred by proposed spatial-provenance embedding protocols and compares it with the baseline scheme LNE \cite{ManishTDSC}.\looseness=-1
\begin{figure}[!h]
    \centering
    \includegraphics[width = 0.25\textwidth]{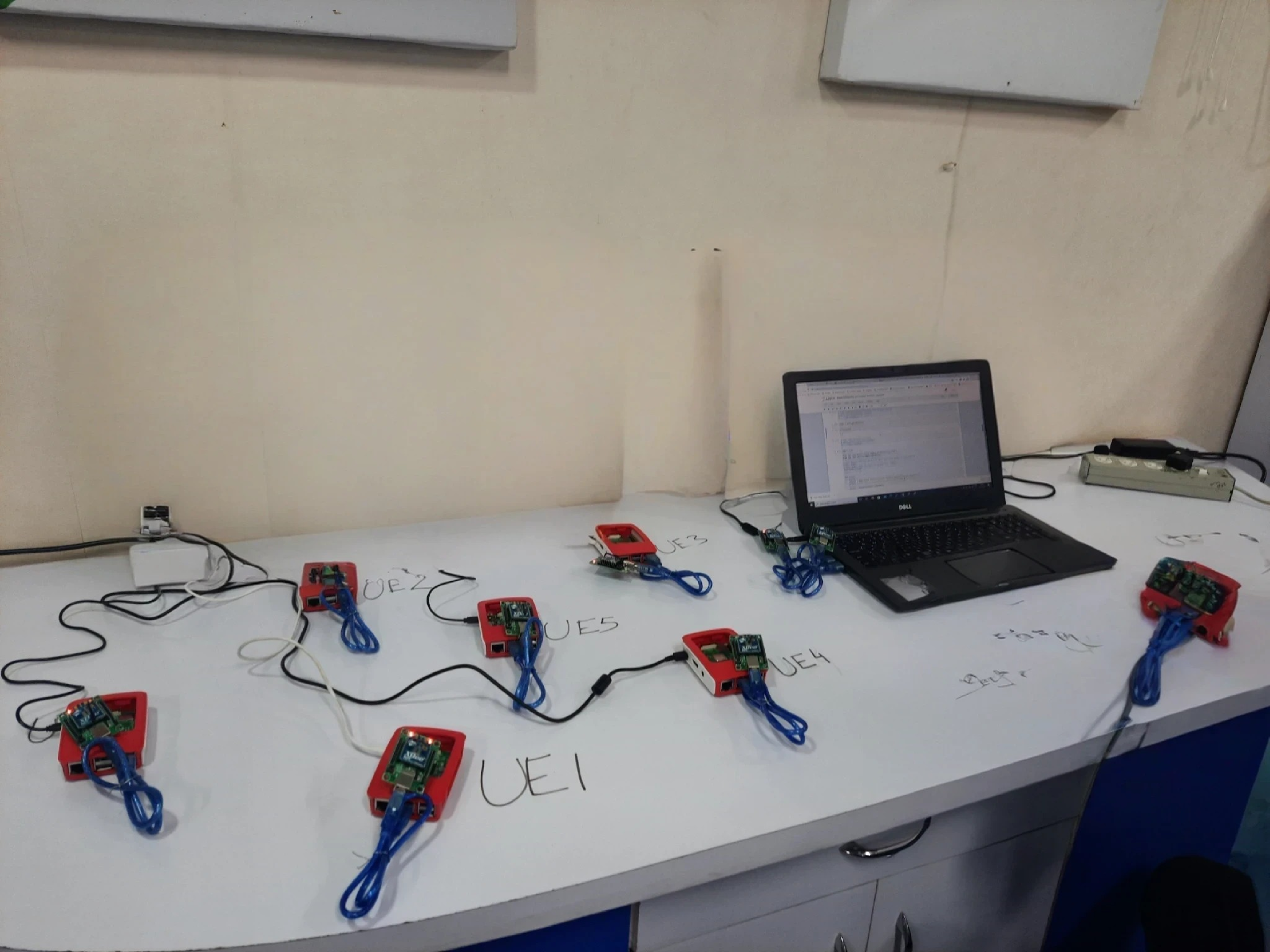}
    \setlength{\belowcaptionskip}{-10pt} 
    \caption{A XBee testbed for BSE and TSE.\looseness=-1}
    \label{fig:XBee}
\end{figure}

For hardware measurements, the testbed system shown in Fig. \ref{fig:XBee} utilizes XBee S2C \cite{XBee_S2C} devices that operate on the ZigBee protocol within the ISM band. We use Raspberry Pis and high-performance computing devices, such as laptops, to meet our computational requirements. A high-performance computing device functions as the RSU and utilizes XBee for wireless communication. Whereas the vehicular node is a combination of a Raspberry Pi interfaced with an XBee. We strategically distribute the vehicular nodes across a specified geographical area to establish the experimental setup.\looseness=-1 

To compare the communication overhead in terms of end-to-end delay for the three methods of spatial-provenance, i.e., LNE \cite{ManishTDSC}, BSE, and TSE, we consider that the packet traverses from the source node to the RSU via an $h$-hop path. In LNE \cite{ManishTDSC}, the deterministic-edge-embedding technique \cite{ManishTDSC, manish_WCNC} encodes the logical path traveled by the packet onto $\textbf{BF}_1$. Whereas, segment ID is embedded onto $\textbf{BF}_2$ by each node along the forwarding path \cite{ManishTDSC, manish_WCNC}. In BSE, as defined in Section \ref{sec:BSE_Embedding}, each forwarding node appends its segment ID to the tail of the packet to share it with an embedding node. Assuming the segment ID occupies 1 byte, this action temporarily increases the packet size by one byte.
When the packet reaches an embedding node, its own segment ID along with the segment ID of the forwarding node is embedded onto $\textbf{BF}_2$, restoring the packet to its original size. This alternating process continues along the path, where a forwarding node increases the size of the packet by one byte, and an embedding node reverts it to its original size. Other processes, like updating the hop counter and $\textbf{BF}_1$, follow the instructions in Section \ref{sec:BSE_Embedding}.\looseness=-1

In TSE, a forwarding node shares its segment ID with an embedding node by appending it to the tail of the packet, similar to BSE described above. Whereas, the next-hop receiver explicitly shares its segment ID with an embedding node through an appropriate segment ID sharing mechanism.
For simplicity, it is assumed that this explicit sharing by next-hop is of size one byte. Other processes, like updating the hop counter and $\textbf{BF}_1$, follow the instructions in Section \ref{sec:TSE_Embedding}.\looseness=-1

For timing analysis, we evaluate the following time components that contribute to the total latency: (1) $TE_{0}$ - time required at a node solely to increment the hop counter; (2) $TE_{1}$ - time required for a node to update the hop counter, compute the Hash values for a single Bloom filter (either $\textbf{BF}_1$ or $\textbf{BF}_2$), and update the corresponding Bloom filter with the relevant Hash values; (3) $TE_{2}$ - time required to increment the hop counter, compute Hash values using $k_1$ and $k_2$ Hash functions, and subsequently update both $\textbf{BF}_1$ and $\textbf{BF}_2$ accordingly; (4) $TP_{x}$ - total time required for a packet of size $x$ bytes to traverse from the application layer of the transmitting node to the application layer of the receiving node. Further, we use $m_1=m_2=30$ bytes and $k_1=k_2=3$. Additionally, headers and other control information add an extra 20 bytes to the payload, bringing the total packet size to 80 bytes. With this packet size, the total end-to-end delay incurred in transmitting a packet across an $h$-hop path for a given protocol can be computed in a straightforward manner.\looseness=-1

Towards computing the above numbers, we executed a large number of iterations on a high-performance computer to ascertain the average values of the time components, which are $TE_{0} \approx 2 \mu s$, $TE_{1} \approx 10 \mu s$, and $TE_{2} \approx 17 \mu s$. We calculated $TP_{x}$ by averaging the results of transmitting a large number of packets between two devices where wireless communication was enabled by XBee S2C devices placed one meter apart. The measured $TP_{x}$ for a packet with size 80 bytes is $TP_{80}=24,840.34 \mu s$, for a packet of size 81 bytes is $TP_{81}=25,046.02 \mu s$, and for a packet with size 21 bytes (for a one-byte packet) is $TP_{21}=17,064.36 \mu s$. By plugging the relevant delay values in the end-to-end delay numbers of LNE, BSE and TSE, the corresponding end-to-end delay in the ZigBee use case is presented as a function of $h$ in Fig. \ref{fig:ZIgBee_total_propagation_time}. It is observed from Fig. \ref{fig:ZIgBee_total_propagation_time} that end-to-end delay for BSE and LNE \cite{ManishTDSC} is similar across varying hop counts, whereas TSE exhibits higher delay. This observation suggests that BSE and LNE \cite{ManishTDSC} outperform TSE in terms of end-to-end delay.  These results also suggest that under similar network conditions, when transmission delay dominates over hardware delay, BSE and LNE \cite{ManishTDSC} are efficient choices. However, BSE may be preferred between the two as it has a lower value of $P_X$ and better reliability than LNE \cite{ManishTDSC}, as demonstrated in Fig. \ref{fig:LN_LSE_LDE_Comparison}.\looseness=-1 


For 5G, assuming a data rate of 100 Mbps and a bandwidth of $20$ MHz, and assuming the distance between adjacent nodes is 1 meter, the total transmission time for a packet of size 80 bytes is $T_{80} = 6.4 \mu s + 3.3ns \approx 6.403 \mu s$. Similarly, $TP_{21} = 1.6833 \mu s$ and $TP_{81} = 6.4833 \mu s$. By plugging the relevant delay values (as defined above) in the end-to-end delay expressions, the corresponding end-to-end delay in the 5G use case is presented as a function of $h$ in Fig. \ref{fig:5G_total_propagation_time}. It is observed from Fig. \ref{fig:5G_total_propagation_time} that BSE and TSE have lower end-to-end delays compared to LNE \cite{ManishTDSC} under identical network parameters in the 5G scenario. This demonstrates the advantage of the skipping technique employed in BSE and TSE, which is highlighted by a significant reduction in transmission delay in 5G networks, making hardware delay the primary contributor to the overall end-to-end delay. TSE is anticipated to outperform both BSE and LNE \cite{ManishTDSC} since every third node performs the embedding in TSE, thereby minimizing hardware delay. However, we observe that TSE performs similarly to BSE, which indicates that the performance of TSE is hampered by the overhead of propagating segment IDs from the next-hop node to the embedding node. TSE may perform better if the segment IDs of neighbor nodes can be shared with the embedding node using a method that incurs less communication overhead. Between BSE and TSE under the same network parameters, when hardware delay dominates transmission delay, TSE is preferable as it offers lower $P_X$ and better reliability (as seen in Fig. \ref{fig:LN_LSE_LDE_Comparison}). In contrast, when transmission delay dominates, BSE is preferable, as it offers a lower $P_X$ over LNE \cite{ManishTDSC}, while maintaining similar delay values. Overall, both BSE and TSE outperform LNE \cite{ManishTDSC} in terms of end-to-end delay and reliability.\looseness=-1
\vspace{-0.3cm}
\begin{figure}[!tbp]
  \begin{subfigure}[b]{0.49\columnwidth}
  \caption{}
    \includegraphics[trim={0cm 0.0cm 1.5cm 0.5cm},clip,scale=0.30]{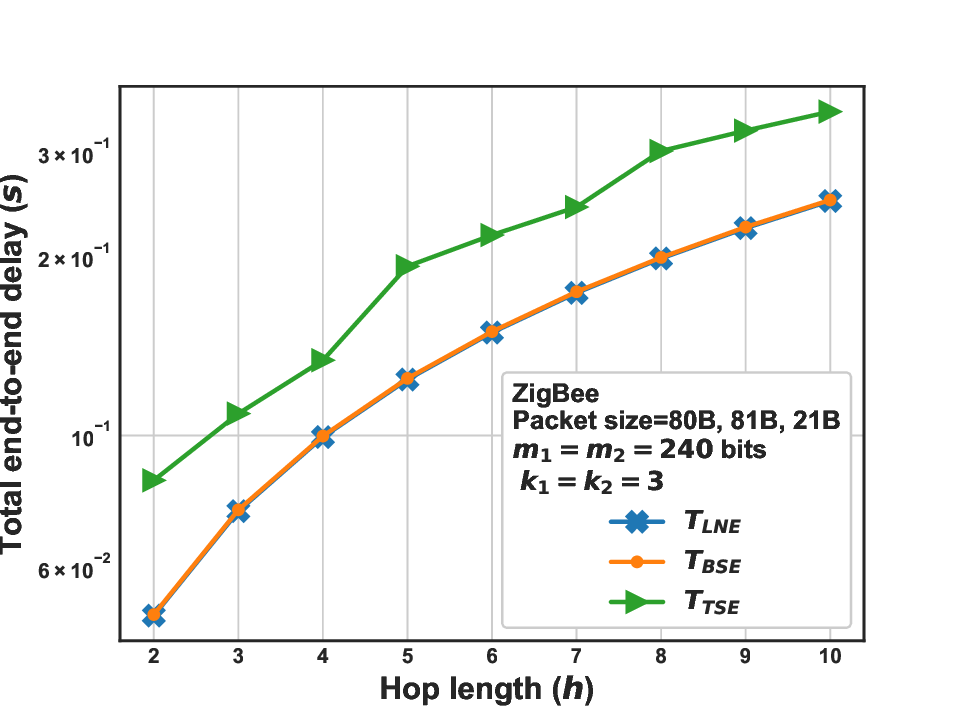}
     \label{fig:ZIgBee_total_propagation_time}
  \end{subfigure}
  \hfill
  \begin{subfigure}[b]{0.49\columnwidth}
  \caption{}
    \includegraphics[trim={0cm 0.0cm 1.5cm 0.5cm},clip,scale=0.30]{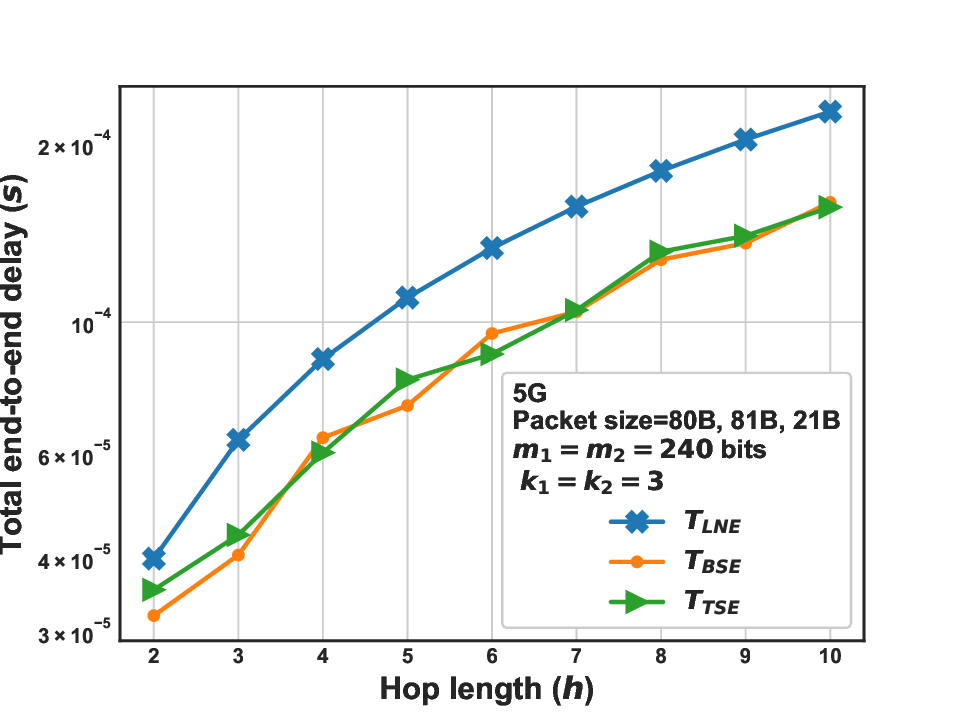}
    \label{fig:5G_total_propagation_time}
  \end{subfigure}
  \vspace{-0.5cm}
  \caption{Comparison of end-to-end delay in: (a) ZigBee (b) 5G use case.\looseness=-1 }
  \vspace{-0.4cm}
\end{figure}
\subsection{Robustness in Dynamic Vehicular Scenarios}\label{sec:mobility analysis}
\bl{To study the impact of mobility and dynamically changing topology, we evaluate our proposed protocols under a realistic vehicular mobility model. For the linear road model considered in Section \ref{sec:network_model}, we assume that the vehicles move in the same direction, i.e., towards the RSU, with constant, though possibly heterogeneous, speeds. In this section, we address two aspects of mobility. The first aspect concerns the effect of vehicle mobility on the frequency of dictionary broadcast and possible changes in privacy levels. The second aspect concerns how mobility impacts the consistent exchange of segment IDs between neighbors in BSE and TSE.\looseness=-1}


\bl{As stated in Section \ref{sec:network_model}, the dictionary $g$ is determined through mutual agreement between the RSU and the participating vehicles. We assume that $g$ remains stable for a considerable period, since in practical networks, vehicles do not change their privacy parameters rapidly. We assume that the RSU periodically broadcast $g$ across its coverage area to help vehicles entering the coverage region, particularly from the farthest segment, obtain the current dictionary and use it to embed the segment ID in packets. If participating vehicles agree to new privacy levels, the RSU recomputes and broadcast the updated dictionary in the next broadcast cycle. Let $\tau_s$ denote the segment retention time, defined as the time a vehicle spends within a single segment at a given speed, $\tau_b$ denote the broadcast interval of the dictionary, and $\tau_{e2e}$ denote the end-to-end delay experienced by a packet while traversing from the source node to the RSU. For the consistent operation of the proposed protocols, we assume that the timing hierarchy $\tau_{e2e} < \tau_b < \tau_s$ is maintained in the network. The condition $\tau_b < \tau_s$ ensures that a vehicle receives the current dictionary before exiting its segment. The validity of  $\tau_b < \tau_s$ under different mobility scenarios has been studied in \cite{ManishTDSC}. The condition $\tau_{e2e} < \tau_b$ ensures that by the time any dictionary update is broadcast, the packet has already reached the RSU, eliminating the risk of mid-path inconsistency in embedding segment ID by the forwarding vehicles.\looseness=-1}

\bl{If $\tau_{e2e} > \tau_b$, a broadcast cycle elapses while a packet is still in transit. This is not a problem in itself because the dictionary does not change with every broadcast. The only problematic case arises when a dictionary change occurs while a packet is in transit, causing some forwarding vehicles to use the updated dictionary while other vehicles on the same path use the previous version of the dictionary. This inconsistency can be resolved by having vehicles maintain two versions of the dictionary and include a small flag in the packet header indicating whether the dictionary was refreshed during the lifetime of the packet. In case where $\tau_b > \tau_s$, a vehicle may exit its segment before receiving a valid dictionary, violating the precondition of our proposed protocols. However, such scenarios are very rare in practical networks. The timing hierarchy $\tau_{e2e} < \tau_b < \tau_s$ is satisfied under typical mobility scenarios, ensuring consistent operation of the proposed protocols.\looseness=-1}

\bl{For the second aspect, several timing components determine how long the embedded segment information remains valid for each packet as it traverses the network. Let $\tau_a$ seconds denote the packet arrival time after a vehicle enters the RSU coverage area. When a packet arrives at the vehicle, a processing delay of $\tau_d$ seconds occurs in embedding or extracting information from the packet before it is forwarded to the next hop. For the packet to reach the next vehicle, the total time taken is $\tau_{tx} +\tau_{prop}$ seconds, where $\tau_{tx} =\frac{8B}{R}$ seconds denotes the transmission time for a packet of size $B$ bytes transmitted at data rate $R$ (bits per second).
The propagation delay $\tau_{prop}$ represents the time taken by the signal to propagate between vehicles. 
Since vehicles continue to move during the combined duration $\tau_a+\tau_d+ \tau_{tx} +\tau_{prop}$ seconds, a neighbor may cross a segment boundary before the packet is fully received and processed, rendering the embedded segment information incorrect.\looseness=-1}

\bl{Using the above timing components, we define a failure event when a packet carrying a segment ID from a forwarding vehicle reaches the embedding vehicle after the forwarding vehicle and/or the embedding vehicle has moved such that the embedded segment IDs violate the $\beta$-hop communication constraint. Henceforth, the probability of occurrence of such events is referred to as failure probability.\looseness=-1}



\bl{In BSE protocol, a failure event occurs if the embedding vehicle crosses a segment boundary before completing packet reception, processing, and embedding. In TSE protocol, the definition of a failure event remains unchanged; however, the effective communication window is larger because TSE introduces an additional interaction phase involving the next-hop vehicle. This additional phase increases the time before the embedded information is finalized, thereby increasing the likelihood of a failure event. In our evaluation, this additional protocol overhead is modeled as a fixed delay that is added to the BSE completion time. Finally, the failure probability is estimated as the fraction of packet dissemination events that violate the segment validity condition described above, while excluding overtaking events, so that only the impact of mobility and timing effects is captured.\looseness=-1}

\bl{For experiments, we consider that the RSU coverage area of 5 km is equally segmented into $r=10$ segments of length $l=500$ m. For BSE, we consider two vehicles initially located within the last two segments, i.e., $A_9$ and $A_{10}$, with initial positions drawn uniformly at random over this region and the embedding vehicle always ahead. The speed of the forwarding vehicle is fixed at 100 kmph, while the speed of the embedding vehicle is varied in the range 50 kmph to 200 kmph, to capture different relative mobility scenarios. The packet of size $B=81$ bytes is transmitted at $R=10$ Mbps, and the processing delay at each vehicle is $\tau_d= 2$ ms. To model the packet arrival, we consider two stochastic processes with identical mean inter-arrival time $1/\lambda$: a Poisson process and a normalized uniform process over $[0,2/\lambda]$. Normalization for uniform process is required to ensure the same average traffic load, which enables a direct comparison of the effect of arrival-time variability. For TSE, an additional delay of $5$ ms is added to model the communication delay incurred between the embedding vehicle and the next-hop vehicle. For a given embedding-vehicle speed and a given packet-arrival rate $\lambda$, we consider $10^6$ packet transmissions. For each packet, we check whether the failure event occurs.
The empirical failure probability is calculated as the ratio of the number of failures to the total number of packets transmitted. For each value of $\lambda$, the resulting failure probabilities are plotted to obtain the curves shown in Fig. \ref{fig:mobility_scenario}.\looseness=-1}

\begin{figure}[ht!]
 \centering

    \includegraphics[trim={0 0 0cm 0cm},clip,scale=0.25]{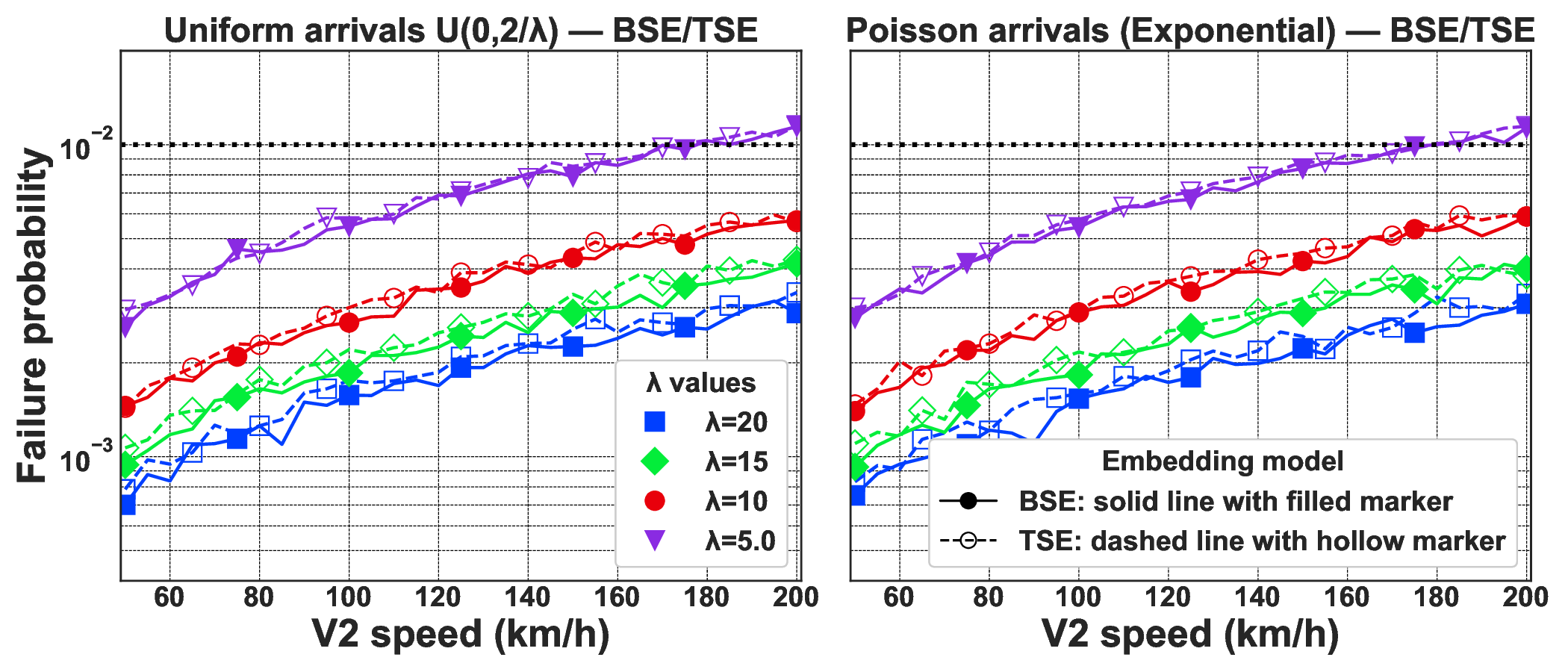}
  \caption{\bl{(a) Failure probability of BSE and TSE under normalized uniform packet arrivals for different values of $\lambda$ (b) Failure probability of BSE and TSE under Poisson (exponential) packet arrivals for the same $\lambda$ values.\looseness=-1}}
  \label{fig:mobility_scenario}
  \vspace{-0.4cm}
\end{figure}

\bl{The results in Fig. \ref{fig:mobility_scenario} show that BSE consistently depicts a lower failure probability than TSE at the same vehicle speeds, while both schemes exhibit similar increasing trends with mobility. The trends observed in both packet arrival models are comparable, indicating that failure probability is primarily affected by timing constraints induced by mobility. Overall, our proposed protocols are robust in dynamic changing topology, making them suitable for practical deployments.\looseness=-1}
\section{Discussion \& Future Direction}\label{discussion}
This study aimed to collectively address the localization requirements of the RSU and the privacy concerns of the vehicles for low-latency V2X networks. \bl{Our proposed protocols can be integrated on top of existing V2X stacks \cite{802.11p, WAVE}, assuming a reliable link layer that subsumes lower-layer mechanisms such as error correction codes, and retransmission schemes. In practice, integrating the proposed framework with PHY/MAC layers introduces additional challenges such as congestion, packet loss, and unreliable neighbor sharing. However, the analysis presented in this work, including the false-positive characterization, the optimization of Bloom filter parameters, and the techniques for choosing Hash functions, remains relevant and applicable in such settings. The only expected consequence of such integration is a marginal increase in end-to-end delay, which is natural. Furthermore, no existing V2X standard \cite{802.11p, WAVE} reserve dedicated bits for provenance. Therefore, for practical implementation within existing standards, the proposed framework must be accommodated within the payload portion of the underlying standard. Recall that the choice of the Bloom filter size in our analysis was chosen arbitrarily, and not aligned with any existing standard. However, our experiments demonstrate that the overheads introduced by inserting the information on location provenance is marginal. For instance, a 100-bit Bloom filter was shown to offer less than 1\% error rate in a 10-hop network while occupying only $5\%$ of a typical V2X payload of 250 bytes, confirming the feasibility of integrating the proposed framework into existing wireless standards.} 

From the perspective of scalability, the error rates of the underlying data structure, i.e., Bloom filters, increase with the number of hops and network size. Therefore, there is a tradeoff of error rates with the increase in packet size overhead. This motivates several directions for future work, such as \bl{studying the joint effect of PHY/MAC layer integration, mapping the proposed framework onto existing wireless standards}, exploring scalability through adaptive Bloom filter sizing, parameter optimization, or hierarchical deployment strategies.\looseness=-1
 

\appendices
\section{Proofs of Main Results}
\subsection{Proof of Proposition \ref{prop:sigma_B}} \label{apx:sigma_B}
\begin{proof}   
All the $h+1$ nodes of an $h$-hop path have associated segment IDs, where subscripts of the segment ID of each node can take values in $[1,r]$ as specified in Section \ref{sec:Routing Constraints}. Assume a valid $h$-hop sequence of subscripts of segment IDs $\mathbf{p}_x = \{p_{x_1}, p_{x_2},\ldots, p_{x_{h+1}}\} \in \mathcal{P}_{\beta}$, where $x \in [1,|\mathcal{P}_{\beta}|]$, $p_{x_z} \in [1,r]~|~ z \in [1,h+1]$, and $p_{x_1}=1$ always, as $p_{x_1}$ denotes the subscript of the segment ID of the RSU. In BSE, $\floor{\frac{h}{2}}$ non-overlapping bi-segment tuples are embedded onto $\textbf{BF}_2$, as segment IDs are embedded by every alternate node in an $h$-hop path. The non-overlapping bi-segment tuples enumerated from the source node to the RSU in $\mathbf{p}_x$ are of the form $\{(p_{x_i},p_{x_{i+1}})~|~ i=h-2t,t \in [0,\floor{\frac{h}{2}}-1]\}$.
Based on Definition \ref{def:bi-segment tuple} of bi-segment tuples, for each value of $p_{x_i} \in [1,r]$, the corresponding $p_{x_{i+1}}$  takes the values in $[p_{x_i}, \min(p_{x_i}+\beta,r)]$. Hence, the number of valid $(p_{x_i},p_{x{i+1}})$ bi-segment tuples for a specific $p_{x_i}$ is given by $(\min(p_{x_i}+\beta,r) -p_{x_i}+1)$. By summing over all the possible values of $p_{x_i}$ from $1$ to $r$, and multiplying by the number of such bi-segment tuples in the $h$-hop path, we arrive at the total number of bi-segment tuples using \eqref{eq:sigma_B}. This completes the proof.\looseness=-1
\end{proof}
\vspace{-0.3cm}
\subsection{Proof of Proposition \ref{prop:sigma_T} } \label{apx:sigma_T}
\begin{proof}  
In TSE, there is $\floor{\frac{h+1}{3}}$ non-overlapping tri-segment tuples in a $h$-hop path, $\mathbf{p}_x \in \mathcal{P}_{\beta}$, because every third node embeds the segment ID (refer Section \ref{sec:TSE_Embedding}). The non-overlapping bi-segment tuples enumerated from the source node to the RSU in $\mathbf{p}_x$ are of the form $\{(p_{x_{i-1}}, p_{x_{i}},p_{x_{i+1}}) ~|~ i =h-2-3t, t \in [0, \floor{\frac{h-2}{3}}]\}$. 
Based on Definition \ref{def:tri-segment tuple} of tri-segment tuple, we observe that for each value of $p_{x_{i}} \in [1,r]$, the corresponding values taken by $p_{x_{i-1}} \in [\max(1, p_{x_{i}}-\beta), p_{x_{i}}]$ and $p_{x_{i+1}} \in [p_{x_{i}}, \min(p_{x_{i}}+\beta, r)]$. The number of such combinations for a given value of $p_{x_{i}}$ is given by $H(p_{x_i})=(p_{x_{i}}-\max(1, p_{x_{i}}-\beta)+1 )( \min(p_{x_{i}}+\beta, r)-p_{x_{i}}+1)$. However, in a special case when $(h+1)~mod~3=0$, in the leftmost tri-segment tuple, i.e., $(p_{x_1},p_{x_2}, p_{x_3})$, the value of $p_{x_1}=1$ as it denotes the subscript of the segment ID of the RSU. Therefore, in this tri-segment tuple, the subscript of the segment ID of the first node is fixed, and the summation for this tri-segment tuple carried over only the values of $p_{x_2}$ and $p_{x_3}$ that satisfy the constraint  $p_{x_2} \in [1, \min(1+\beta,r)]$ and for a given value of $p_{x_2}$, the value of $p_{x_3} \in [p_{x_2}, \min(p_{x_2}, r)]$. Therefore, the total number of valid tri-segment tuples of segment ID, denoted as $\sigma_{T}$, for an $h$-hop path is given by \eqref{eq:sigma_T}. This completes the proof.\looseness=-1
\end{proof}
\vspace{-0.3cm}
\subsection{Proof of Theorem \ref{Theorem-Main_eq}} \label{apx:Theorem-Main_eq}
\begin{proof}
When a packet traverses a $h$-hop path, there will be $h^*_X$ embeddings in $\textbf{BF}_2$. However, due to the characteristics of the Bloom filter, let us assume that there are $j$ extra recoveries from $\textbf{BF}_2$ as described in Definition \ref{def:setR}. The probability of retrieving $j$ elements from $\mathcal{F}_X$ while not recovering the remaining $(|\mathcal{F}_X|-j)$ elements is expressed as $p_{1}^{j} p_{2}^{|\mathcal{F}_X|-j}$.
Because of $j$ extra recoveries, a specific valid segment sequence $\mathbf{p}_x, ~x \in [1, |\mathcal{P}_{\beta}|]$, may result into one or more additional segment sequence such as $\mathbf{p}_{x_1}, \mathbf{p}_{x_2}, \ldots,\mathbf{p}_{x_{C_{x,{j_X}}}}$, i.e., $E_{X}^{(2)}|\alpha,\mathbf{p}_x=\mathbf{p}_{x_1} \bigcup \mathbf{p}_{x_2}\bigcup \ldots\bigcup\mathbf{p}_{x_{C_{x,{j_X}}}}$. However, the appearance of all the additional segment sequences is an independent and mutually exclusive event. Therefore, $\Pr(E_{X}^{(2)}|\alpha,\mathbf{p}_x)=\Pr(\mathbf{p}_{x_1}) + \Pr(\mathbf{p}_{x_2})+\ldots+\Pr(\mathbf{p}_{x_{C_{x,{j_X}}}})$.
Using Definition \ref{def:collision}, Definition \ref{def:setF} and Definition \ref{def:setR}, the probability of false positive event can be expressed as\looseness=-1

\begin{small}\begin{IEEEeqnarray*}{rcl}
\Pr(E_{X}^{(2)}|\alpha, g)= \Pr(E_{X}^{(2)}|\alpha, \mathbf{p}_{x}) = \sum_{{j}=1}^{|\mathcal{F}_X|} p_1^{j}\times p_2^{|\mathcal{F}_X|-{j}}\times C_{X,x,j}.
\end{IEEEeqnarray*}\end{small}

\noindent \noindent Thus, the false positive probability averaged over all possible valid segment sequences can be computed as\looseness=-1
\begin{small}\begin{IEEEeqnarray}{rcl}
\label{overall_fp}
\frac{1}{|\mathcal{P}_{\beta}|}\sum_{{j}=1}^{|\mathcal{F}_X|} p_1^{j}\times p_2^{|\mathcal{F}_X|-{j}}\times \left( \sum_{x = 1}^{|\mathcal{P}_{\beta}|}C_{X,x,j}\right),
\end{IEEEeqnarray}\end{small}
where uniform distribution is assumed on $\mathbf{p}_{x}$. We can express the above equation with $\mathcal{G}_{\beta}$ as\looseness=-1 
\begin{small}\begin{IEEEeqnarray}{rcl}
\label{overall_fp2}
\frac{1}{|\mathcal{G}_{\beta}|}\sum_{{j}=1}^{|\mathcal{F}_X|} p_1^{j}\times p_2^{|\mathcal{F}_X|-{j}}\times (\sum_{x = 1}^{|\mathcal{G}_{\beta}|}C_{X,x,j}).
\end{IEEEeqnarray}\end{small}\looseness=-1
\end{proof}
\vspace{-0.6cm}
\subsection{Proof of Proposition \ref{prop:value_of_C_Xx1}} \label{apx:prop:value CXx1}
\begin{proof}
For a valid $h$-hop $\mathbf{p}_x,~x\in[1,|\mathcal{P}_{\beta}|]$, one can compute $C_{B,x,1}$ by considering a subsequence of $\mathbf{p}_x$ as $(p_{x_{i-1}}, p_{x_{i}}, p_{x_{i+1}}, p_{x_{i+2}}) ~|~ i \in \{2,3,\ldots,h-2\}$. For each pair position $i \in \{2,3,\ldots,h-2\}$ in $\mathbf{p}_x$, count the number of alternative values that can be assigned to pair $(p_{x_{i}}, p_{x_{i+1}})$ while rest of the $\mathbf{p}_x$ remains unchanged and valid. The values of $p_{x_{i}}$ and $p_{x_{i+1}}$ can only change within the bounds defined by $p_{x_{i-1}}$ and $p_{x_{i+2}}$, while still satisfying the $\beta$ constraint between consecutive elements. The value of $C_{B,x,1}$ is computed by counting all such permissible modifications across the entire $\mathbf{p}_x$.\looseness=-1

To compute $C_{T,x,1}$, follow the same approach as defined above for BSE, with the only difference that three consecutive elements $(p_{x_{i}}, p_{x_{i+1}}, p_{x_{i+2}})$  will be evaluated with in the subsequence $(p_{x_{i-1}}, p_{x_{i}}, p_{x_{i+1}}, p_{x_{i+2}},p_{x_{i+3}})$ of $\mathbf{p}_x$.\looseness=-1
\end{proof}
\vspace{-0.5cm}
\subsection{Proof of Proposition \ref{prop:value extension any beta}} \label{apx:prop:value extension any beta}
\begin{proof}
This equation is derived by sequentially choosing one of the $j$ segment tuples and subsequently selecting the remaining $j-1$ segment tuples from the remaining false segment tuples in $\mathcal{F}_X$ to form the set $\mathcal{R}_X$. Subsequently, two segment tuples from the $j$ segment tuples are fixed, after which the remaining $j-2$ segment tuples can be selected from the remaining false segment tuples of $\mathcal{F}_X$ to form the set $\mathcal{R}_X$. This procedure is reiterated until $j$ segment tuples are fixed to fill $\mathcal{R}_X$, after which the remaining false segment tuples are selected from $\mathcal{F}_X$.\looseness=-1
\end{proof}
\bibliographystyle{IEEEtran} 
\bibliography{refs} 

\noindent \textbf{Manish Bansal} is a Ph.D Research Scholar at IIT Delhi. His research interests include 5G wireless security, provenance, sensor networks, and IoT.


\noindent \textbf{Harshan Jagadeesh} is an Associate Professor at the Department of Electrical Engineering, IIT Delhi. His research interests are in the broad areas of network security, information theory, and coding theory.
\end{document}